\documentclass[10pt,amssymb,amsfont,a4paper]{article}
\usepackage{latexsym,amssymb,amsmath,amsthm,color}
\usepackage{geometry}
\usepackage{url}
\usepackage{graphicx}
\usepackage[utf8]{inputenc}
\usepackage[numbers]{natbib}
\usepackage{authblk}
\usepackage{tikz}
\usepackage{algorithmic}
\usepackage{listings}
\usetikzlibrary{decorations.pathreplacing}
\usetikzlibrary{shapes.misc, fit}

\DeclareMathOperator{\Rk}{Rk}
\DeclareMathOperator{\diag}{Diag}

\DeclareMathOperator{\dd}{d}
\DeclareMathOperator{\wt}{wt}

\theoremstyle{definition}
\newtheorem{definition}{Definition}
\newtheorem{construction}{Construction}
\newtheorem{example}{Example}

\theoremstyle{plain}
\newtheorem{theorem}{Theorem}
\newtheorem{proposition}{Proposition}
\newtheorem{lemma}{Lemma}
\newtheorem{remark}{Remark}
\newtheorem{corollary}{Corollary}

\title{Locally Repairable Convolutional Codes with \\Sliding Window Repair}
\author[1]{Umberto Mart{\'i}nez-Pe\~{n}as \thanks{umberto@ece.utoronto.ca; umberto.martinez@unine.ch}}
\author[2]{Diego Napp \thanks{diego.napp@ua.es}}
\affil[1]{Dept.\ of Electrical \& Computer Engineering,
University of Toronto, Canada}
\affil[2]{Dept.\ of Mathematics, University of Alicante, Spain}

\date{}

\begin{document}

\lstset{numbers=left, numberstyle=\tiny, stepnumber=1, numbersep=5pt}

\maketitle

\begin{abstract}
Locally repairable convolutional codes (LRCCs) for distributed storage systems (DSSs) are introduced in this work. They enable local repair, for a single node erasure (or more generally, $ \partial - 1 $ erasures per local group), and sliding-window global repair, which can correct erasure patterns with up to $ \dd^c_j - 1 $ erasures in every window of $ j+1 $ consecutive blocks of $ n $ nodes, where $ \dd^c_j $ is the $ j $th column distance of the code. The parameter $ j $ can be adjusted, for a fixed LRCC, according to different catastrophic erasure patterns, requiring only to contact $ n(j+1) - \dd^c_j + 1 $ nodes, plus less than $ \mu n $ other nodes, in the storage system, where $ \mu $ is the memory of the code. A Singleton-type bound is provided for $ \dd^c_j $. If it attains such a bound, an LRCC can correct the same number of catastrophic erasures in a window of length $ n(j+1) $ as an optimal locally repairable block code of the same rate and locality, and with block length $ n(j+1) $. In addition, the LRCC is able to perform the flexible and somehow local sliding-window repair by adjusting $ j $. Furthermore, by adjusting and/or sliding the window, the LRCC can potentially correct more erasures in the original window of $ n(j+1) $ nodes than an optimal locally repairable block code of the same rate and locality, and length $ n(j+1) $. Finally, the concept of partial maximum distance profile (partial MDP) codes is introduced. Partial MDP codes can correct all information-theoretically correctable erasure patterns for a given locality, local distance and information rate. An explicit construction of partial MDP codes whose column distances attain the provided Singleton-type bound, up to certain parameter $ j=L $, is obtained based on known maximum sum-rank distance convolutional codes.


\end{abstract}

\section{Introduction} \label{sec intro}

Locally repairable codes (LRCs) \cite{gopalan} are an important class of codes for Distributed Storage Systems (DSSs), since they allow to repair a single node by contacting and downloading the content of a small number (called \textit{locality}) of other nodes (in contrast with MDS codes), while still being able to repair a large number of nodes in case of catastrophic erasures (in contrast with Cartesian products). LRCs are thus natural hybrids between MDS codes and Cartesian products of codes that enjoy both global and local erasure-correction capabilities simultaneously, given by global and local distances, respectively. Note that \textit{repair} is typically used interchangeably with \textit{erasure correction} in the storage literature. We will use both terms throughout this work.

LRCs have already been implemented in practice (see \cite{azure, xoring} for instance). \textit{Optimal} LRCs (meaning LRCs attaining optimal global distance for a given locality, local distance and information rate) for general parameters and field sizes that are linear in the code length were first obtained in \cite{tamo-barg}. LRCs capable of correcting all information-theoretically correctable global erasure patterns, for a given locality, local distance and information rate, were introduced in \cite{blaum-RAID, gopalan-MR} (where they are called \textit{partial MDS} and \textit{maximally recoverable} LRCs, respectively). As expected, maximally recoverable LRCs also attain optimal global distance. However, they can correct strictly more global erasure patterns than general optimal LRCs (see Remark \ref{remark optimal LRCs are not MR LRCs}) for the same parameters. Constructions of maximally recoverable LRCs with relatively small field sizes for general parameters have been given in \cite{gabrys, lrc-sum-rank} (see also the references therein). 

On the other hand, it is shown in \cite{erasure-convolutional} that maximum distance profile (MDP) convolutional codes provide an interesting alternative to MDS block codes since they admit \textit{sliding-window} erasure correction: They can correct any erasure pattern such that there are no more than $ (n-k)(j+1) $ erasures in any consecutive $ j+1 $ blocks of $ n $ symbols (that is, $ n(j+1) $ consecutive symbols), where $ k/n $ is the rate of the code (see \cite[Th. 3.1]{erasure-convolutional} or Fig. \ref{fig sliding window}). Furthermore, the correction is performed somehow locally by sliding recursively the window of $ j+1 $ blocks, and the parameter $ j $ may vary arbitrarily up to a certain constant $ L $ determined by the degree (thus memory) of the convolutional code (see (\ref{eq def of L})). Therefore MDP convolutional codes already enable certain local and flexible repair, since the window size $ n(j+1) $ can be chosen according to how catastrophic the erasure pattern is. Moreover, by adjusting and/or sliding a window of $ j + 1 $ blocks (see Fig. \ref{fig moving back}), an MDP code can potentially correct in a window of size $ n(j+1) $ more erasures than an MDS block code of the same rate and of block length $ n(j+1) $. Unfortunately, in case of one single node erasure (most common case), sliding-window repair with $ j = 0 $ still requires contacting and downloading the content of $ \mu n $ extra symbols, where $ \mu $ is the memory of the code, due to its convolutional nature.

Motivated by the discussion in the previous paragraph, we introduce in this work \textit{locally repairable convolutional codes} (LRCCs). When being optimal in terms of global distance or maximal recoverability, LRCCs can repair a single node (or more generally, $ \partial - 1 $ erasures per local group) by contacting $ r < n $ (or even $ r < k $) other nodes and simultaneously enable sliding-window repair (see Fig. \ref{fig sliding window repair}), which can be set up flexibly according to different catastrophic erasure patterns (see Figs. \ref{fig sliding window repair} and \ref{fig sliding window repair in partial MDP}), and which can potentially correct in a window of size $ n(j+1) $ more erasures than an optimal or maximally recoverable locally repairable block code of the same rate and locality, and of block length $ n(j+1) $ (see Fig. \ref{fig moving back}).

LRCCs also enable encoding and storing an unrestricted sequence of files, while locality remains constant and encoding and sliding-window repair complexities are all bounded (by the memory of the code). Furthermore, LRCCs can easily be turned into block codes by converting them into tail-biting convolutional codes, while the properties described above still hold.

We now illustrate with Example \ref{example intro} and Fig. \ref{fig moving back} the main advantages of LRCCs over block LRCs. For fairness, we compare optimal LRCCs with optimal LRCs. 

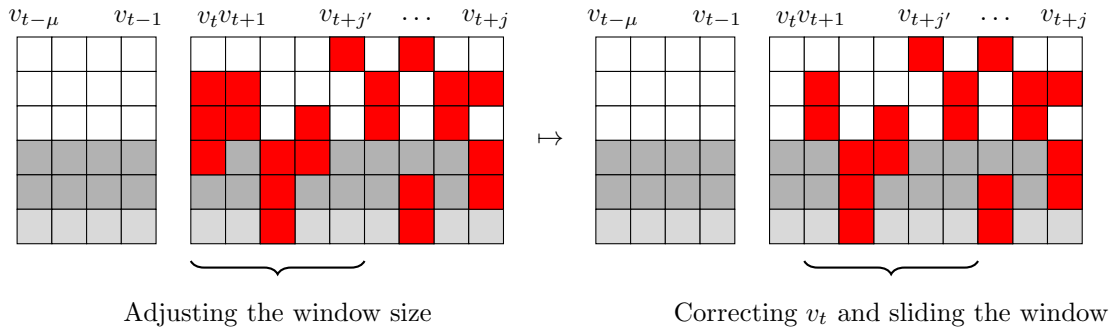
\begin{figure*} [!t]
\begin{center}
\begin{tabular}{c@{\extracolsep{1cm}}c}
\begin{tikzpicture}[
square/.style = {draw, rectangle, 
                 minimum size=\m, outer sep=0, inner sep=0, font=\small,
                 },
                        ]
\def\m{13pt}
\def\w{26} 
\def\wm{25}
\def\h{6}
\def\glob{3}
\def\loc{5}
    \pgfmathsetmacro\uw{int(\w/2)}
    \pgfmathsetmacro\uh{int(\h/2)}

\node [] at (7*\m + \m, 0.01) {$ v_{t - \mu} $};
\node [] at (10*\m + \m, 0.01) {$ v_{t - 1} $};

  \foreach \x in {8,...,11}
  {
    \foreach \y in {1,...,\h}
       {    
           
           \ifnum\y>\loc
               \node [square, fill=gray!30]  (\x,\y) at (\x*\m,-\y*\m) {$  $};
           \else
               \ifnum\y>\glob
                   \node [square, fill=gray!60]  (\x,\y) at (\x*\m,-\y*\m) {$  $};
               \else
                   \node [square, fill=white]  (\x,\y) at (\x*\m,-\y*\m) {$  $};
               \fi
           \fi
       }
  }

\node [] at (12*\m + \m, 0.01) {$ v_{t} $};
\node [] at (13*\m + \m, 0.01) {$ v_{t + 1} $};
\node [] at (16*\m + \m, 0.01) {$ v_{t + j^\prime} $};
\node [] at (18*\m + \m, 0.01) {$ \ldots $};
\node [] at (20*\m + \m, 0.01) {$ v_{t + j} $};

  \foreach \x in {13,...,21}
  {
    \foreach \y in {1,...,\h}
       {    
           
           \ifnum\y>\loc
               \node [square, fill=gray!30]  (\x,\y) at (\x*\m,-\y*\m) {$  $};
           \else
               \ifnum\y>\glob
                   \node [square, fill=gray!60]  (\x,\y) at (\x*\m,-\y*\m) {$  $};
               \else
                   \node [square, fill=white]  (\x,\y) at (\x*\m,-\y*\m) {$  $};
               \fi
           \fi
       }
  }

\node [square, fill=red!100]  (13,2) at (13*\m,-2*\m) {$  $};   
\node [square, fill=red!100]  (13,3) at (13*\m,-3*\m) {$  $};   
\node [square, fill=red!100]  (13,4) at (13*\m,-4*\m) {$  $};  
  
\node [square, fill=red!100]  (14,2) at (14*\m,-2*\m) {$  $};   
\node [square, fill=red!100]  (14,3) at (14*\m,-3*\m) {$  $};  

\node [square, fill=red!100]  (15,4) at (15*\m,-4*\m) {$  $}; 
\node [square, fill=red!100]  (15,5) at (15*\m,-5*\m) {$  $}; 
\node [square, fill=red!100]  (15,6) at (15*\m,-6*\m) {$  $}; 

\node [square, fill=red!100]  (16,3) at (16*\m,-3*\m) {$  $}; 
\node [square, fill=red!100]  (16,4) at (16*\m,-4*\m) {$  $};  

\node [square, fill=red!100]  (17,1) at (17*\m,-1*\m) {$  $};
   
\node [square, fill=red!100]  (18,2) at (18*\m,-2*\m) {$  $};   
\node [square, fill=red!100]  (18,3) at (18*\m,-3*\m) {$  $}; 

\node [square, fill=red!100]  (19,1) at (19*\m,-1*\m) {$  $};
\node [square, fill=red!100]  (19,5) at (19*\m,-5*\m) {$  $};
\node [square, fill=red!100]  (19,6) at (19*\m,-6*\m) {$  $}; 

\node [square, fill=red!100]  (20,2) at (20*\m,-2*\m) {$  $};   
\node [square, fill=red!100]  (20,3) at (20*\m,-3*\m) {$  $};  
   
\node [square, fill=red!100]  (21,2) at (21*\m,-2*\m) {$  $};  
\node [square, fill=red!100]  (21,4) at (21*\m,-4*\m) {$  $};   
\node [square, fill=red!100]  (21,5) at (21*\m,-5*\m) {$  $};

   \draw [draw, decorate, thick,decoration={brace,amplitude=5pt,mirror}]
   (12.5*\m,-7*\m) -- (17.5*\m,-7*\m) node[midway,yshift=-2em]{Adjusting the window size};

\end{tikzpicture}

\begin{tikzpicture}[
square/.style = {draw, rectangle, 
                 minimum size=\m, outer sep=0, inner sep=0, font=\small,
                 },
                        ]
\def\m{13pt}
\def\w{26} 
\def\wm{25}
\def\h{6}
\def\glob{3}
\def\loc{5}
    \pgfmathsetmacro\uw{int(\w/2)}
    \pgfmathsetmacro\uh{int(\h/2)}

\node [] at (5.2*\m + \m, -3.5*\m) {$ \mapsto $};

\node [] at (7*\m + \m, 0.01) {$ v_{t - \mu} $};
\node [] at (10*\m + \m, 0.01) {$ v_{t - 1} $};

  \foreach \x in {8,...,11}
  {
    \foreach \y in {1,...,\h}
       {    
           
           \ifnum\y>\loc
               \node [square, fill=gray!30]  (\x,\y) at (\x*\m,-\y*\m) {$  $};
           \else
               \ifnum\y>\glob
                   \node [square, fill=gray!60]  (\x,\y) at (\x*\m,-\y*\m) {$  $};
               \else
                   \node [square, fill=white]  (\x,\y) at (\x*\m,-\y*\m) {$  $};
               \fi
           \fi
       }
  }

\node [] at (12*\m + \m, 0.01) {$ v_{t} $};
\node [] at (13*\m + \m, 0.01) {$ v_{t + 1} $};
\node [] at (16*\m + \m, 0.01) {$ v_{t + j^\prime} $};
\node [] at (18*\m + \m, 0.01) {$ \ldots $};
\node [] at (20*\m + \m, 0.01) {$ v_{t + j} $};

  \foreach \x in {13,...,21}
  {
    \foreach \y in {1,...,\h}
       {    
           
           \ifnum\y>\loc
               \node [square, fill=gray!30]  (\x,\y) at (\x*\m,-\y*\m) {$  $};
           \else
               \ifnum\y>\glob
                   \node [square, fill=gray!60]  (\x,\y) at (\x*\m,-\y*\m) {$  $};
               \else
                   \node [square, fill=white]  (\x,\y) at (\x*\m,-\y*\m) {$  $};
               \fi
           \fi
       }
  }

\node [square, fill=red!100]  (14,2) at (14*\m,-2*\m) {$  $};   
\node [square, fill=red!100]  (14,3) at (14*\m,-3*\m) {$  $};  

\node [square, fill=red!100]  (15,4) at (15*\m,-4*\m) {$  $}; 
\node [square, fill=red!100]  (15,5) at (15*\m,-5*\m) {$  $}; 
\node [square, fill=red!100]  (15,6) at (15*\m,-6*\m) {$  $}; 

\node [square, fill=red!100]  (16,3) at (16*\m,-3*\m) {$  $}; 
\node [square, fill=red!100]  (16,4) at (16*\m,-4*\m) {$  $};   

\node [square, fill=red!100]  (17,1) at (17*\m,-1*\m) {$  $};
   
\node [square, fill=red!100]  (18,2) at (18*\m,-2*\m) {$  $};   
\node [square, fill=red!100]  (18,3) at (18*\m,-3*\m) {$  $}; 

\node [square, fill=red!100]  (19,1) at (19*\m,-1*\m) {$  $};
\node [square, fill=red!100]  (19,5) at (19*\m,-5*\m) {$  $};
\node [square, fill=red!100]  (19,6) at (19*\m,-6*\m) {$  $}; 

\node [square, fill=red!100]  (20,2) at (20*\m,-2*\m) {$  $};   
\node [square, fill=red!100]  (20,3) at (20*\m,-3*\m) {$  $};  
   
\node [square, fill=red!100]  (21,2) at (21*\m,-2*\m) {$  $};    
\node [square, fill=red!100]  (21,4) at (21*\m,-4*\m) {$  $};   
\node [square, fill=red!100]  (21,5) at (21*\m,-5*\m) {$  $};

   \draw [draw, decorate, thick,decoration={brace,amplitude=5pt,mirror}]
   (13.5*\m,-7*\m) -- (18.5*\m,-7*\m) node[midway,yshift=-2em]{Correcting $ v_t $ and sliding the window};

\end{tikzpicture}

\end{tabular}
\end{center}

\caption{Illustration of Example \ref{example intro}. The window consisting of the $ j+1 = 9 $ blocks $ v_t, v_{t+1}, \ldots, v_{t+j} $ contains $ 21 $ erasures. Hence it cannot be corrected by an LRCC with optimal $ j $th column distance, which is $ 18 $, by considering that window. Furthermore, it could not be corrected either by a block LRC of length $ (j+1)n = 54 $, dimension $ (j+1)k = 27 $, locality $ r = 5 $ and local distance $ \partial - 1 = 2 $, since its distance is also $ 18 $. Assume the LRCC also has optimal $ j^\prime $th column distance for $ j^\prime = 4 $, which would be $ 12 $. Then the LRCC may correct the erasure pattern by reducing the window length to $ j^\prime+1 = 5 $ blocks.  }
\label{fig moving back}
\end{figure*}

\begin{example} \label{example intro}

Consider a $ (6,3) $ convolutional code that encodes a stream of file vectors over a finite field $ \mathbb{F} $, each of length $ k = 3 $, into a stream of encoded vectors, each of length $ n = 6 $. 

Assume a node in the storage system stores a symbol over $ \mathbb{F} $, and call \textit{block} each set of $ n = 6 $ coordinates supporting each encoded vector. In this example, each block forms a \textit{local group}. If the code has locality $ r = 5 $ and local distance $ \partial = 2 $ (Section \ref{sec locality}), it means that a single node erasure ($ \partial - 1 $ node erasures) in each block may be repaired by only contacting the other $ 5 $ nodes in that block. 

The code can correct erasure patterns with up to $ {\rm d}_j^c - 1 $ erasures in every window of $ j+1 $ consecutive blocks of $ n $ symbols, where $ j = 0,1,2, \ldots $ can be adjusted. Assume that an erasure pattern as in Fig. \ref{fig moving back} occurs, with $ 21 $ erasures in a given window of $ j + 1 = 9 $ blocks. If the LRCC has optimal $ j $th column distance $ {\rm d}_j^c $ (as in Corollaries \ref{cor j-MSRD applied to const 1} and \ref{cor construction of partial MDP}), then $ {\rm d}_j^c = 18 $, for $ j = 8 $, and the code cannot correct such erasures considering such a window. Furthermore, an optimal block LRC with block length $ (j+1)n = 54 $, dimension $ (j+1)k = 27 $, locality $ r = 5 $ and local distance $ \partial = 2 $, also has global distance $ 18 $ (see \cite[Eq. (1)]{gopalan}), hence it cannot correct that erasure pattern either.

However, we may adjust the window for the LRCC. Consider instead windows of length $ j^\prime + 1 = 5 $, as in Fig. \ref{fig moving back}. If the LRCC also has optimal $ j^\prime $th column distance (as in Corollaries \ref{cor j-MSRD applied to const 1} and \ref{cor construction of partial MDP}), then $ {\rm d}_{j^\prime}^c = 12 $ for $ j^\prime = 4 $. Observe that now every window of $ j^\prime + 1 = 5 $ consecutive blocks of $ n $ symbols contains at most $ 11 $ erasures. Therefore, the LRCC may correct such an erasure pattern by sliding the new adjusted window of length $ j^\prime + 1 = 5 $, whereas the optimal block LRC as in the previous paragraph cannot.

The disadvantage of the LRCC is that, in order to perform such an erasure correction, we need to read the content (which needs to be correct) of the $ \mu $ blocks of $ n $ symbols previous to such window (see Fig. \ref{fig moving back}), where $ \mu $ is the memory of the LRCC. 
\end{example}

\begin{example} \label{example intro 2}
Consider now a $ (6,4) $ LRCC with locality $ r = 5 $ and local distance $ \partial = 2 $ (Section \ref{sec locality}). Assume also that the code has memory $ \mu = 5 $ and degree $ \delta = 20 $ (Subsection \ref{subsec basic definitions}) and that it is an optimal LRCC (as in Corollaries \ref{cor j-MSRD applied to const 1} and \ref{cor construction of partial MDP}). As in Example \ref{example intro}, such an LRCC can correct the same number of erasures in any window consisting of $ L+1 = 26 $ consecutive blocks ($ n(L+1) = 156 $ symbols) as an optimal block LRC of the same rate ($ 2/3 $), same locality ($ r = 5 $), same local distance ($ \partial = 2 $) and total length $ n(L+1) = 156 $. In both cases, such a number of erasures is $ 32 = (n-k)(L+1) - \left\lceil \frac{k(L+1)}{r} \right\rceil + 1 $ (see Theorem \ref{th singleton bound on column distances} and \cite[Eq. (1)]{gopalan}, respectively). 

However, the LRCC may correct any erasure pattern with up to $ (n-k)(j+1) - \left\lceil \frac{k(j+1)}{r} \right\rceil + 1 $ erasures in any window of $ n(j+1) $ consecutive nodes, for all $ j = 0,1,\ldots, L = 25 $. This only requires reading and downloading the content of the remaining nodes in that window (that is, $ k(j+1) + \left\lceil \frac{k(j+1)}{r} \right\rceil - 1 $ symbols), plus another $ \mu N = \mu (n-1) = 25 $ previous nodes (see Fig. \ref{fig sliding window repair in partial MDP}). For instance, for $ j = 2 $, we may repair any erasure pattern with up to $ 4 $ erasures in any consecutive $ 18 $ nodes ($ 3 $ blocks of $ 6 $ nodes), by contacting $ 14 $ nodes in that window, plus another $ \mu N = 25 $ previous nodes. The optimal LRC, in contrast, would always require contacting $ 124 = k(L+1) + \left\lceil \frac{k(L+1)}{r} \right\rceil - 1 $ other nodes in order to repair any $ 4 $ erasures in $ 3 $ blocks of $ 6 $ nodes, since at least one of these blocks contains $ 2 $ erasures, which cannot be repaired locally. 

Therefore, adjusting the window size when using an LRCC also reduces the number of nodes that need to be contacted. Thus sliding-window erasure correction of LRCCs provides a type of erasure correction in between local and global erasure correction.
\end{example}

Our main contributions are the following. We define LRCCs (Definition \ref{def lrc convolutional codes}) and provide a Singleton-type bound on their column distances (Theorem \ref{th singleton bound on column distances}), which measure the global sliding window repair capability of the code. We later define partial MDP codes (Definition \ref{def partial MDP}), which can correct all information-theoretically correctable erasure patterns for the given local constrains and, in particular, attain the previous bound for as long as possible. We provide in Construction \ref{construction} a method for finding partial MDP codes based on outer MSRD convolutional codes (Theorem \ref{th MSRD implies partial MDS}). By plugging in Construction \ref{construction} the MSRD convolutional codes from \cite{mahmood-convolutional}, we obtain an explicit family of partial MDP codes (Corollary \ref{cor construction of partial MDP}) for general parameters. Their main disadvantage is their big global field size, although local fields are small. However, this is only an issue in terms of computational complexity, since nodes in DSSs typically store large amounts of data. Furthermore, our construction gives some field size to guarantee the existence of partial MDP codes, but constructions over smaller fields may be possible. 

To conclude this introduction, we note that the use of streaming or convolutional codes for storage or as LRCs is not new. Binary tail-biting convolutional codes were proposed as LRCs in \cite{lrc-convoluted, rcc}, but sliding-window repair was not considered. Locality properties of more general (but still binary) convolutional codes were recently considered in \cite{ivanov}. However, LRCCs and sliding-window repair as considered in this work were not treated in \cite{ivanov}. Rateless streaming codes (e.g. Fountain codes \cite{fountain}) are an interesting alternative to MDS block codes for global repair in DSSs (see \cite[Ch. 50]{mackay}), since they generally achive low redundancy and enable global erasure correction with complexity of $ \mathcal{O}(k \log(k)) $ XOR operations (products in $ \mathbb{F}_2 $) or even less, for $ k $ encoded symbols. Locally repairable Fountain codes were proposed in \cite{fountainLRC}. However, their locality is of order $ \log(k) $ (unbounded), for $ k $ encoded symbols, and they do not enable sliding-window global repair.

The remainder of the paper is organized as follows. In Section \ref{sec preliminaries}, we collect some preliminaries on convolutional codes. In Section \ref{sec locality}, we introduce LRCCs and give a Singleton-type bound on their column distances, which determine the sliding-window erasure-correction capability of LRCCs. In Section \ref{sec connecting LRCC and SR convolutional}, we show how to obtain LRCCs with arbitrary and small-field local codes and optimal global column distances (in view of the previous bound) based on codes with optimal column sum-rank distances \cite{mahmood-convolutional}. In Section \ref{sec partial MDS convolutional}, we introduce \textit{partial MDP} convolutional codes, whose sliding windows can correct analogous erasure patterns as partial MDS block codes \cite{blaum-RAID, gopalan-MR}. We also provide concrete constructions of partial MDP convolutional codes based on the codes in \cite{mahmood-convolutional}. Finally, in Section \ref{sec further considerations}, we discuss extending our work to considering LRCCs with unequal localities and local distances, and how to turn our LRCCs to tail-biting convolutional codes.

\section{Preliminaries on Convolutional Codes} \label{sec preliminaries}

In this section, we collect general definitions and results on convolutional codes that we will use throughout the paper.

Let $ \mathbb{F} $ be a finite field, and denote by $ \mathbb{F}[D] $ the ring of polynomials with coefficients in $ \mathbb{F} $. Fix a positive integer $ n \in \mathbb{N} $. We will typically consider and graphically represent a word in $ \mathbb{F}[D]^n $ as an unrestricted sequence of vectors of length $ n $, $ v(D) = \sum_{j \in \mathbb{N}} v_j D^j \equiv (v_0, v_1, v_2, \ldots) \in (\mathbb{F}^n)^\mathbb{N} $, where we use the following terminology. A \textit{block} is each of the $ n $ consecutive coordinates in $ (\mathbb{F}^n)^\mathbb{N} $ that support each vector $ v_0, v_1, \ldots $, being \textit{the $ j $th block} the block containing the coordinates supporting $ v_j $, for $ j \in \mathbb{N} $. A \textit{symbol} is each component of the vectors $ v_0, v_1, \ldots $, thus it is an element of $ \mathbb{F} $. Finally, a \textit{node} is the abstraction of the storage device that stores a given symbol. Hence, in this work, each block corresponds to $ n $ nodes storing $ n $ symbols over $ \mathbb{F} $.

\subsection{Degree and Memory} \label{subsec basic definitions}

Recall that, since $ \mathbb{F}[D] $ is a principal ideal domain, every $ \mathbb{F}[D] $-submodule of $ \mathbb{F}[D]^n $ is free.

\begin{definition} \label{def convolutional codes}
An $ (n,k) $ \textit{convolutional code} is a (free) $ \mathbb{F}[D] $-submodule $ \mathcal{C} \subseteq \mathbb{F}[D]^n $ of rank $ k $. A \textit{generator matrix} of the code is a full-rank matrix $ G(D) \in \mathbb{F}[D]^{k \times n} $ such that
$$ \mathcal{C} = \left\lbrace u(D) G(D) \mid u(D) \in \mathbb{F}[D]^k \right\rbrace . $$
For a vector $ v(D) \in \mathbb{F}[D]^n $, we define its degree as the maximum degree of its components, which are polynomials in $ \mathbb{F}[D] $. We say that a generator matrix $ G(D) $ of $ \mathcal{C} $ is \textit{reduced} if the sum of the row degrees of $ G(D) $ is minimum among generator matrices of $ \mathcal{C} $, where by row degrees we mean the degrees of the rows in $ G(D) $. 
\end{definition}

It follows from Theorem A-2, Item 3, in \cite{general-mceliece} that if $ e_1 \leq e_2 \leq \ldots \leq e_k $ and $ f_1 \leq f_2 \leq \ldots \leq f_k $ are the row degrees of a reduced generator matrix $ G(D) \in \mathbb{F}[D]^{k \times n} $ and some other generator matrix $ \widetilde{G}(D) \in \mathbb{F}[D]^{k \times n} $, respectively, of $ \mathcal{C} $, then $ e_i \leq f_i $, for $ i = 1,2, \ldots, k $. In particular, the set of degrees $ \{ e_1, e_2, \ldots, e_k \} $ of one, thus any, reduced generator matrix is an invariant of the convolutional code $ \mathcal{C} $. Hence the following definition is consistent.

\begin{definition}
Given an $ (n,k) $ convolutional code $ \mathcal{C} \subseteq \mathbb{F}[D]^n $, let $ e_1, e_2, \ldots, e_k $ be the row degrees of one, thus any, of its reduced generator matrices. We define the \textit{degree} and \textit{memory} of $ \mathcal{C} $, respectively, as
$$ \delta = \delta(\mathcal{C}) = \sum_{i=1}^k e_i \quad \textrm{and} \quad \mu = \mu(\mathcal{C}) = \max \{ e_1, e_2, \ldots, e_k \}. $$
\end{definition}

Note that convolutional codes with zero memory (thus zero degree) coincide with (potentially infinite) Cartesian products of a single $ (n,k) $ block code $ \mathcal{C} \subseteq \mathbb{F}^n $.

\subsection{Non-Catastrophic Codes and Parity-Check Matrices} \label{subsec non-catastrophic}

In most results in this work, although not all, we will require convolutional codes to be \textit{non-catastrophic} or \textit{observable}, which we now define in terms of basic generator matrices. 

\begin{definition} \label{def basic and non-catastrophic}
Given an $ (n,k) $ convolutional code $ \mathcal{C} \subseteq \mathbb{F}[D]^n $, we say that a generator matrix $ G(D) $ of $ \mathcal{C} $ is basic if it has a polynomial right inverse, that is, if there exists $ F(D) \in \mathbb{F}[D]^{n \times k} $ such that $ G(D) F(D) = I_k $. We say that $ \mathcal{C} $ is non-catastrophic if it admits a generator matrix that is reduced and basic.
\end{definition}

Observe that any reduced and basic generator matrix $ G(D) = \sum_{j=0}^\mu G_j D^j $ of a convolutional code satisfies that $ G_0 \in \mathbb{F}^{k \times n} $ is full-rank. For many results in this work, we will only need this weaker property. 

Using Theorem A-1, Item 5, in \cite{general-mceliece}, and using the vector space over $ \mathbb{F}(D) $ (the field of fractions of $ \mathbb{F}[D] $) generated by a non-catastrophic convolutional code, it is easy to see that it admits a polynomial parity-check matrix. This strong property of non-catastrophic codes is what we will need for sliding-window repair, as described in Subsection \ref{subsec sliding window}.

\begin{lemma} \label{lemma parity check}
If $ \mathcal{C} \subseteq \mathbb{F}[D]^n $ is a non-catastrophic $ (n,k) $ convolutional code, then there exists a full-rank matrix $ H(D) \in \mathbb{F}[D]^{(n-k) \times n} $ such that
$$ \mathcal{C} = \{ v(D) \in \mathbb{F}[D]^n \mid v(D) H(D)^T = 0 \}. $$
We call $ H(D) $ a (polynomial) \textit{parity-check matrix} of $ \mathcal{C} $.
\end{lemma}

\subsection{Free and Column Distances} \label{subsec free and column dist}

We now recall the main notions of minimum distance of convolutional codes. Given $ v(D) = \sum_{j \in \mathbb{N}} v_j D^j \in \mathbb{F}[D]^n $, we define its \textit{Hamming weight} as
$$ \wt(v(D)) = \sum_{j \in \mathbb{N}} \wt(v_j). $$

The free distance, which we now define, gives the correction capability of a convolutional code when considering whole codewords. In other words, there is no maximum degree $ j $ for a codeword considered by the free distance.

\begin{definition} \label{def free distance}
Given an $ (n,k) $ convolutional code $ \mathcal{C} \subseteq \mathbb{F}[D]^n $, we define its \textit{free distance} as
\begin{equation*}
\dd(\mathcal{C}) = \min \{ \wt(v(D)) \mid v(D) \in \mathcal{C} \textrm{ and } v(D) \neq 0 \} .
\end{equation*}
\end{definition}

We next define column distances, which give the sliding-window correction capability of a non-catastrophic convolutional code (see the next subsection). This will be the type of distance that we will be interested in for global repair in our locally repairable codes.

\begin{definition} \label{def column distance}
Given an $ (n,k) $ convolutional code $ \mathcal{C} \subseteq \mathbb{F}[D]^n $, with memory $ \mu $ and reduced generator matrix $ G(D) = \sum_{h=0}^\mu G_h D^h $, define the $ j $th \textit{truncated sliding generator matrix} $ G_j^c \in  \mathbb{F}^{(j+1)k \times (j+1)n} $ as
$$ G_j^c = \left[ \begin{array}{cccc}
G_0 & G_1 & \ldots & G_j \\
 & G_0 & \ldots & G_{j-1} \\
 & & \ddots & \vdots \\
 & & & G_0
\end{array} \right], $$
for $ j \in \mathbb{N} $, where $ G_h = 0 $ if $ h > \mu $. Define now the $ j $th \textit{column block code} of $ \mathcal{C} $ as
\begin{equation*}
\begin{split}
\mathcal{C}_j^c = & \left\lbrace (u_0, u_1, \ldots, u_j)G_j^c \mid u_0, u_1, \ldots, u_j \in \mathbb{F}^k, u_0 \neq 0 \right\rbrace \\
 \stackrel{*}{=} & \left\lbrace (v_0, v_1, \ldots, v_j) \mid \sum_{h \in \mathbb{N}} v_h D^h \in \mathcal{C}, v_0 \neq 0 \right\rbrace \subseteq \mathbb{F}^{(j+1)n},
\end{split}
\end{equation*}
where the equality $ * $ holds if $ G_0 $ is full-rank. Finally, define the \textit{$ j $th column distance} of $ \mathcal{C} $ as
$$ \dd_j^c(\mathcal{C}) =  \min \left\lbrace {\rm wt}_H(v) \mid v \in \mathcal{C}_j^c \right\rbrace , $$
where note that $ v \neq 0 $ if $ v \in \mathcal{C}_j^c $, for $ j \in \mathbb{N} $.
\end{definition}

The column distances satisfy the following Singleton bound, which was proven in \cite[Prop. 2.2]{stronglyMDS}.

\begin{proposition}[\textbf{\cite{stronglyMDS}}] \label{prop singleton column distance}
For an $ (n,k) $ convolutional code $ \mathcal{C} \subseteq \mathbb{F}[D]^n $ with a generator matrix $ G(D) = \sum_{j=0}^\mu G_jD^j $ (possibly not reduced) such that $ G_0 $ is full-rank, and for $ j \in \mathbb{N} $, it holds that
\begin{equation}
\dd_j^c(\mathcal{C}) \leq (n-k)(j+1) + 1.
\label{eq classical singleton bound on column distances}
\end{equation}
\end{proposition}

Items 1 and 2 in the following proposition follow from \cite[Cor. 2.3]{stronglyMDS} and \cite[Th. 2.2]{MDSconvolutional}, respectively. 

\begin{proposition}[\textbf{\cite{stronglyMDS, MDSconvolutional}}] \label{prop monotonicity singleton column distance}
Given a non-catastrophic $ (n,k) $ convolutional code $ \mathcal{C} \subseteq \mathbb{F}[D]^n $ of degree $ \delta $, the following hold:
\begin{enumerate}
\item
If $ \dd_{j+1}^c(\mathcal{C}) = (n-k)(j+2) + 1 $, then $ \dd_j^c(\mathcal{C}) = (n-k)(j+1) + 1 $.
\item
If $ \dd_j^c(\mathcal{C}) = (n-k)(j+1) + 1 $, then
\begin{equation}
j \leq L = \left\lfloor \frac{\delta}{k} \right\rfloor + \left\lfloor \frac{\delta}{n-k} \right\rfloor.
\label{eq def of L}
\end{equation}
\end{enumerate}
\end{proposition}

The previous proposition motivates the following definition.

\begin{definition} \label{def mds and mdp}
We say that an $ (n,k) $ convolutional code $ \mathcal{C} $ is $ j $-MDS if it is non-catastrophic and $ \dd_j^c(\mathcal{C}) = (n-k)(j+1) + 1 $. We say that $ \mathcal{C} $ is \textit{maximum distance profile (MDP)} if it is non-catastrophic and $ \dd_L^c(\mathcal{C}) = (n-k)(L+1) + 1 $, where $ L $ is as in (\ref{eq def of L}).
\end{definition}

\subsection{Sliding-Window (Global) Repair} \label{subsec sliding window}

As shown in \cite[Th. 3.1]{erasure-convolutional} and its proof, a non-catastrophic $ (n,k) $ convolutional code $ \mathcal{C} \subseteq \mathbb{F}[D]^n $ may correct any erasure pattern with up to $ \dd^c_j(\mathcal{C}) - 1 $ erasures in any tuple $ (v_t, v_{t+1}, \ldots, v_{t+j}) \in \mathbb{F}^{n(j+1)} $, for $ j \in \mathbb{N} $. Furthermore, it may do so recursively by sliding a window that only involves the symbols in $ v_{t-\mu}, v_{t-\mu+1}, \ldots, v_{t+j} $, where $ \mu = \mu(\mathcal{C}) $. The formal statement is as follows. See also Fig. \ref{fig sliding window} for a graphical description.

For convenience, we first define erasures formally. 

\begin{definition} \label{def erasures}
Let $ \star $ be a symbol not belonging to any finite field, and denote $ \widetilde{\mathbb{F}} = \mathbb{F} \cup \{ \star \} $. Given $ N \in \mathbb{N} $ and $ v \in \mathbb{F}^N $, we say that $ v^* \in \widetilde{\mathbb{F}}^N $ is the vector $ v $ with $ e $ \textit{erasures}, where $ 0 \leq e \leq N $, if $ e $ components of $ v^* $ are the symbol $ \star $, and $ v $ and $ v^* $ coincide in the other $ N-e $ components. 
\end{definition} 

We now state \cite[Th. 3.1]{erasure-convolutional} and part of its proof. 

\begin{theorem} [\textbf{\cite{erasure-convolutional}}] \label{th sliding window}
Let $ \mathcal{C} \subseteq \mathbb{F}[D]^n $ be a non-catastrophic $ (n,k) $ convolutional code with memory $ \mu $, and fix $ j \in \mathbb{N} $. Let $ v(D) = \sum_{h \in \mathbb{N}} v_h D^h \in \mathcal{C} $ and let $ v^*_0, v^*_1, v^*_2, \ldots \in \widetilde{\mathbb{F}}^n $ be such that $ (v_t^*,v_{t+1}^*, \ldots, v_{t+j}^*) \in \widetilde{\mathbb{F}}^{n(j+1)} $ is the vector $ (v_t,v_{t+1}, \ldots, v_{t+j}) \in \mathbb{F}^{n(j+1)} $ with at most $ \dd^c_j(\mathcal{C}) - 1 $ erasures, for all $ t \in \mathbb{N} $. Then, for each $ t = 0,1,2, \ldots $, the vector $ v_{t} \in \mathbb{F}^n $ can be recursively and uniquely recovered from the tuple 
\begin{equation}
(v_{t-\mu}, v_{t-\mu+1}, \ldots, v_{t-1}, v^*_{t}, v^*_{t+1}, \ldots, v^*_{t+j}) \in \widetilde{\mathbb{F}}^{n(\mu + j + 1)} 
\label{eq window vector}
\end{equation}
by solving a system of non-homogeneous equations, whose coefficients are given by a parity-check matrix of $ \mathcal{C} $ (Lemma \ref{lemma parity check}), the symbols in $ v_{t-\mu}, v_{t-\mu+1}, \ldots, v_{t-1} $, and the symbols such that $ v_{u,i}^* = v_{u,i} $, for $ u = t, t+1, \ldots, t + j $, and whose unknowns are $ x_{u,i} $, for $ i $ such that $ v_{u,i}^* = \star $, for $ u = t,t+1, \ldots, t + j $. To recover $ v_{t+1} $, we ``slide'' the window (\ref{eq window vector}) one position to the right (see Fig. \ref{fig sliding window}).
\end{theorem}

In the previous theorem, we implicitly assume that $ v_j = 0 $ for all $ j = -1,-2, \ldots, -\mu $.

This type of erasure correction may already be considered as local repair, since $ j $ may be small. Furthermore, the window size is not necessarily restricted, since $ j $ may be arbitrary. However, setting $ j=0 $, we see that correcting one erasure in a single block $ v_t \in \mathbb{F}^n $ requires contacting another $ \mu n $ nodes and downloading their symbols, corresponding to $ (v_{t-\mu}, v_{t-\mu+1}, \ldots, v_{t-1}) \in \mathbb{F}^{\mu n} $, in order to set up the necessary system of linear equations. Thus, although sliding-window repair enjoys certain local nature, it admits considerable room for improvement. Adding locality inside each block $ v_t $ optimally will be our main objective in the rest of the paper.

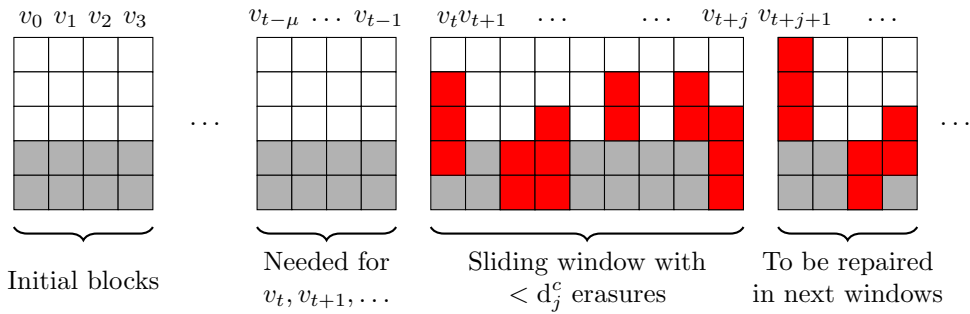
\begin{figure*} [!t]
\begin{center}
\begin{tabular}{c@{\extracolsep{1cm}}c}
\begin{tikzpicture}[
square/.style = {draw, rectangle, 
                 minimum size=\m, outer sep=0, inner sep=0, font=\small,
                 },
                        ]
\def\m{13pt}
\def\w{26} 
\def\wm{25}
\def\h{5}
\def\glob{3}

\foreach \x in {0,...,3}
{
  \node [] at (\x*\m + \m, 0.01) {$ v_{\x} $};
}

  \foreach \x in {1,...,4}
  {
    \foreach \y in {1,...,\h}
       {     
               \ifnum\y>\glob
                   \node [square, fill=gray!60]  (\x,\y) at (\x*\m,-\y*\m) {$  $};
               \else
                   \node [square, fill=white]  (\x,\y) at (\x*\m,-\y*\m) {$  $};
               \fi
       }
  }

\node [] (1, \h) at (5*\m + 1*\m, -3*\m) {$ \ldots $};

\node [] at (7*\m + \m, 0.01) {$ v_{t - \mu} $};
\node [] at (8.5*\m + \m, 0.01) {$ \ldots $};
\node [] at (10*\m + \m, 0.01) {$ v_{t - 1} $};

  \foreach \x in {8,...,11}
  {
    \foreach \y in {1,...,\h}
       {    
               \ifnum\y>\glob
                   \node [square, fill=gray!60]  (\x,\y) at (\x*\m,-\y*\m) {$  $};
               \else
                   \node [square, fill=white]  (\x,\y) at (\x*\m,-\y*\m) {$  $};
               \fi
       }
  }

\node [] at (12*\m + \m, 0.01) {$ v_{t} $};
\node [] at (13*\m + \m, 0.01) {$ v_{t + 1} $};
\node [] at (15*\m + \m, 0.01) {$ \ldots $};
\node [] at (18*\m + \m, 0.01) {$ \ldots $};
\node [] at (20*\m + \m, 0.01) {$ v_{t + j} $};

  \foreach \x in {13,...,21}
  {
    \foreach \y in {1,...,\h}
       {    
               \ifnum\y>\glob
                   \node [square, fill=gray!60]  (\x,\y) at (\x*\m,-\y*\m) {$  $};
               \else
                   \node [square, fill=white]  (\x,\y) at (\x*\m,-\y*\m) {$  $};
               \fi
       }
  }

\node [square, fill=red!100]  (13,2) at (13*\m,-2*\m) {$  $};   
\node [square, fill=red!100]  (13,3) at (13*\m,-3*\m) {$  $};   
\node [square, fill=red!100]  (13,4) at (13*\m,-4*\m) {$  $};  

\node [square, fill=red!100]  (15,4) at (15*\m,-4*\m) {$  $}; 
\node [square, fill=red!100]  (15,5) at (15*\m,-5*\m) {$  $}; 

\node [square, fill=red!100]  (16,3) at (16*\m,-3*\m) {$  $}; 
\node [square, fill=red!100]  (16,4) at (16*\m,-4*\m) {$  $};   
\node [square, fill=red!100]  (16,5) at (16*\m,-5*\m) {$  $};
   
\node [square, fill=red!100]  (18,2) at (18*\m,-2*\m) {$  $};   
\node [square, fill=red!100]  (18,3) at (18*\m,-3*\m) {$  $}; 

\node [square, fill=red!100]  (20,2) at (20*\m,-2*\m) {$  $};   
\node [square, fill=red!100]  (20,3) at (20*\m,-3*\m) {$  $};  
   
\node [square, fill=red!100]  (21,3) at (21*\m,-3*\m) {$  $}; 
\node [square, fill=red!100]  (21,4) at (21*\m,-4*\m) {$  $};   
\node [square, fill=red!100]  (21,5) at (21*\m,-5*\m) {$  $};

\node [] at (24.5*\m + \m, 0.01) {$ \ldots $};  
\node [] at (22*\m + \m, 0.01) {$ v_{t + j + 1} $};  

  \foreach \x in {23,...,26}
  {
    \foreach \y in {1,...,\h}
       {    
               \ifnum\y>\glob
                   \node [square, fill=gray!60]  (\x,\y) at (\x*\m,-\y*\m) {$  $};
               \else
                   \node [square, fill=white]  (\x,\y) at (\x*\m,-\y*\m) {$  $};
               \fi
       }
  }

\node [] (\w + 1, \h) at (\w*\m + \m + 0.6*\m, -3*\m) {$ \ldots $};

\node [square, fill=red!100]  (23,1) at (23*\m,-1*\m) {$  $}; 
\node [square, fill=red!100]  (23,2) at (23*\m,-2*\m) {$  $};   
\node [square, fill=red!100]  (23,3) at (23*\m,-3*\m) {$  $};      
   
\node [square, fill=red!100]  (25,4) at (25*\m,-4*\m) {$  $};   
\node [square, fill=red!100]  (25,5) at (25*\m,-5*\m) {$  $};    

\node [square, fill=red!100]  (26,3) at (26*\m,-3*\m) {$  $};   
\node [square, fill=red!100]  (26,4) at (26*\m,-4*\m) {$  $};

   \draw [draw, decorate, thick,decoration={brace,amplitude=5pt,mirror}]
   (0.5*\m,-6*\m) -- (4.5*\m,-6*\m) node[midway,yshift=-2em]{Initial blocks};        
   \draw [draw, decorate, thick,decoration={brace,amplitude=5pt,mirror}]
   (7.5*\m,-6*\m) -- (11.5*\m,-6*\m) node[midway,yshift=-2em]{\begin{tabular}{c}Needed for\\ $ v_t, v_{t+1}, \ldots $ \end{tabular}};       
   \draw [draw, decorate, thick,decoration={brace,amplitude=5pt,mirror}]
   (12.5*\m,-6*\m) -- (21.5*\m,-6*\m) node[midway,yshift=-2em]{\begin{tabular}{c}Sliding window with\\ $ < \dd^c_j $ erasures \end{tabular}};        
   \draw [draw, decorate, thick,decoration={brace,amplitude=5pt,mirror}]
   (22.5*\m,-6*\m) -- (26.5*\m,-6*\m) node[midway,yshift=-2em]{\begin{tabular}{c}To be repaired\\ in next windows \end{tabular}};

\end{tikzpicture}

\end{tabular}
\end{center}

\caption{Sliding-window erasure correction as described in Theorem \ref{th sliding window} for a $ (5,3) $ convolutional code. Here, parity-check symbols are depicted in grey, and erasures are depicted in red.}
\label{fig sliding window}
\end{figure*}

\section{Locality in Convolutional Codes} \label{sec locality}

In this section, we formulate locality for convolutional codes. For this purpose, we define the following two types of restrictions for a convolutional code. The first type consists in considering one generic block $ v_j \in \mathbb{F}^n $ for arbitrary codewords $ v(D) $ in the convolutional code.

\begin{definition} \label{def restriction C^0}
Given an $ (n,k) $ convolutional code $ \mathcal{C} \subseteq \mathbb{F}[D]^n $ with reduced generator matrix $ G(D) = \sum_{j=0}^\mu G_j D^j $, where $ \mu $ is the memory of the code, we define its \textit{associated block code} as
\begin{equation*}
\begin{split}
\mathcal{C}^0 = & \left\lbrace \sum_{j=0}^\mu u_j G_j \mid u_j \in \mathbb{F}^k, j = 0,1,\ldots, \mu \right\rbrace \\
 = & \left\lbrace v_\mu \in \mathbb{F}^n \mid v(D) = \sum_{j \in \mathbb{N}} v_j D^j \in \mathcal{C} \right\rbrace \subseteq \mathbb{F}^n.
\end{split}
\end{equation*}
\end{definition}

Note that by the second equality, the definition of $ \mathcal{C}^0 $ does not depend on the generator matrix of $ \mathcal{C} $. We now give the second type of restriction, which consists in restricting each block of the convolutional code to some subset of coordinates $ \Gamma \subseteq [n] $. Here, we use the notation $ [n] = \{ 1,2, \ldots, n \} $.

\begin{definition} \label{def restriction C_Gamma}
Given an $ (n,k) $ convolutional code $ \mathcal{C} \subseteq \mathbb{F}[D]^n $ and given a non-empty subset $ \Gamma \subseteq [n] $, we define the \textit{restriction} of $ \mathcal{C} $ to $ \Gamma $ as the convolutional code
$$ \mathcal{C}_\Gamma = \left\lbrace v(D)_\Gamma \mid v(D) \in \mathcal{C} \right\rbrace \subseteq \mathbb{F}[D]^{|\Gamma|} . $$
Here,  if $ v \in \mathbb{F}^n $, then $ v_\Gamma \in \mathbb{F}^{|\Gamma|} $ denotes the projection of $ v $ onto the coordinates in $ \Gamma $. Then if $ v(D) = \sum_{j \in \mathbb{N}} v_j D^j \in \mathbb{F}[D]^n $, we use the notation $ v(D)_\Gamma = \sum_{j \in \mathbb{N}} (v_j)_\Gamma D^j \in \mathbb{F}[D]^{|\Gamma|} $.

For a matrix $ G(D) \in \mathbb{F}[D]^{k \times n} $, we denote by $ G(D)_\Gamma \in \mathbb{F}[D]^{k \times |\Gamma|} $ the matrix whose rows are the rows of $ G(D) $ restricted to $ \Gamma $. 
\end{definition}

Observe that if $ G(D) \in \mathbb{F}[D]^{k \times n} $ is a generator matrix of $ \mathcal{C} $, then the rows of $ G(D)_\Gamma \in \mathbb{F}^{k \times |\Gamma|} $ generate $ \mathcal{C}_\Gamma $, although they may not be $ \mathbb{F}[D] $-linearly independent.

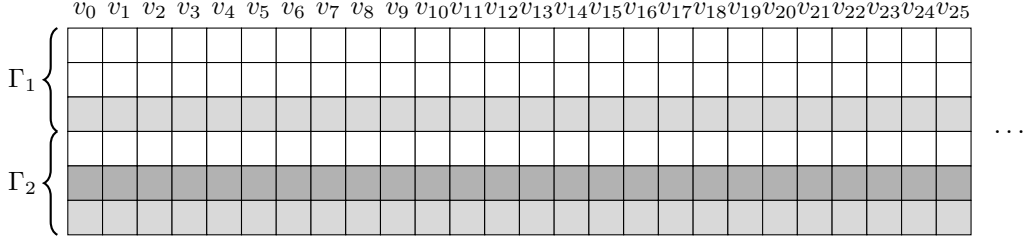
\begin{figure*} [!t]
\begin{center}
\begin{tabular}{c@{\extracolsep{1cm}}c}
\begin{tikzpicture}[
square/.style = {draw, rectangle, 
                 minimum size=\m, outer sep=0, inner sep=0, font=\small,
                 },
                        ]
\def\m{13pt}
\def\w{26} 
\def\wm{25}
\def\h{6}
\def\glob{4}
\def\locf{3}
\def\locs{5}
    \pgfmathsetmacro\uw{int(\w/2)}
    \pgfmathsetmacro\uh{int(\h/2)}

\foreach \x in {0,...,\wm}
{
  \node [] at (\x*\m + \m, 0.01) {$ v_{\x} $};
}

  \foreach \x in {1,...,\w}
  {
    \foreach \y in {1,...,\h}
       {    
           
           \ifnum\y>\locs
               \node [square, fill=gray!30]  (\x,\y) at (\x*\m,-\y*\m) {$  $};
           \else
               \ifnum\y=\locf
                   \node [square, fill=gray!30]  (\x,\y) at (\x*\m,-\y*\m) {$  $};
               \else                   
                   \ifnum\y>\glob
                       \node [square, fill=gray!60]  (\x,\y) at (\x*\m,-\y*\m) {$  $};
                   \else
                       \node [square, fill=white]  (\x,\y) at (\x*\m,-\y*\m) {$  $};
                   \fi
               \fi    
           \fi
       }
  }
   
   \draw [decorate, thick,decoration={brace,amplitude=5pt}]
   (0.2*\m,-3.5*\m) -- (0.2*\m,-0.5*\m) node[midway,xshift=-1.3em]{$ \Gamma_1 $};    
   \draw [decorate, thick,decoration={brace,amplitude=5pt}]
   (0.2*\m,-6.5*\m) -- (0.2*\m,-3.5*\m) node[midway,xshift=-1.3em]{$ \Gamma_2 $};

\node [] (\w + 1, \h) at (\w*\m + \m + 0.6*\m, -3.5*\m) {$ \ldots $};
        
\end{tikzpicture}

\end{tabular}
\end{center}

\caption{Graphical description of an $ (n,k,r,\partial) = (6,3,2,2) $ LRCC with $ 2 $ local groups, $ \Gamma_1 $ and $ \Gamma_2 $, each of size $ 3 $ and constituting the two halves of each column. White, light grey and dark grey boxes depict information symbols, local parities and global parities, respectively. Typically in this case ($ \partial = 2 $), local parities may be chosen as the XOR of the other symbols in the local group (see Construction \ref{construction}). }
\label{fig LRCC}
\end{figure*}

We may now extend the definition of $ (r,\partial) $-locality for block codes from \cite[Def. 1]{kamath} to convolutional codes.

\begin{definition} \label{def lrc convolutional codes}
We say that an $ (n,k) $ convolutional code $ \mathcal{C} \subseteq \mathbb{F}[D]^n $ has \textit{$ (r,\partial) $-locality} if there exist non-empty sets $ \Gamma_i $, for $ i = 1,2, \ldots, g $, such that $ [n] = \bigcup_{i=1}^g \Gamma_i $, and
\begin{enumerate}
\item
$ |\Gamma_i| \leq r + \partial - 1 $,
\item
$ \dd(\mathcal{C}_{\Gamma_i}^0) \geq \partial $,
\end{enumerate}
for $ i = 1,2, \ldots, g $. Here, we write $ \mathcal{C}_{\Gamma_i}^0 $ instead of $ (\mathcal{C}_{\Gamma_i})^0 = (\mathcal{C}^0)_{\Gamma_i} $ for simplicity. Thus, $ \mathcal{C}_{\Gamma_i}^0 $ denotes the block code associated (Definition \ref{def restriction C^0}) to the restriction (Definition \ref{def restriction C_Gamma}) of $ \mathcal{C} $ on $ \Gamma_i $.

We say then that $ \mathcal{C} $ is an \textit{$ (n,k,r,\partial) $ locally repairable convolutional code}, or LRCC for short. The set $ \Gamma_i $ is called the \textit{$ i $th local group}, for $ i = 1,2,\ldots, g $, and $ r $ and $ \partial $ are called the \textit{locality} and \textit{local distance} of $ \mathcal{C} $, respectively.
\end{definition}

In other words, we consider local groups in each block of $ n $ symbols, corresponding to terms $ v_j \in \mathbb{F}^n $ in a codeword $ v(D) = \sum_{j \in \mathbb{N}} v_j D^j \in \mathcal{C} $. See Fig. \ref{fig LRCC} for a graphical example of a $ (6,3,2,2) $ LRCC with $ 2 $ local groups. In contrast to block codes, local repair with only one local group ($ g = 1$) per block already outperforms sliding-window repair even when $ j=0 $, in terms of total contacted nodes, see Fig. \ref{fig sliding window repair}.

We state now the local erasure-correction capability of LRCCs. Definition \ref{def lrc convolutional codes} is given so that the following result holds. The proof is straightforward.
\begin{proposition}
Let $ \mathcal{C} \subseteq \mathbb{F}[D]^n $ be an $ (n,k,r,\partial) $ LRCC with local groups $ \Gamma_i $, for $ i = 1,2, \ldots, g $. Fix $ j \in \mathbb{N} $ and $ i = 1,2,\ldots, g $. For all $ v(D) = \sum_{j \in \mathbb{N}} v_jD^j \in \mathcal{C} $, if $ v^* \in \widetilde{\mathbb{F}}^{|\Gamma_i|} $ is the vector $ (v_j)_{\Gamma_i} \in \mathbb{F}^{|\Gamma_i|} $ with at most $ \partial - 1 $ erasures (see Definition \ref{def erasures}), then we may uniquely recover the vector $ (v_j)_{\Gamma_i} $ from $ v^* $ by using the restricted block code $ \mathcal{C}_{\Gamma_i}^0 \subseteq \mathbb{F}^{|\Gamma_i|} $, without contacting nodes or reading symbols outside of $ \Gamma_i $ in the $ j $th block of the convolutional code. 
\end{proposition}

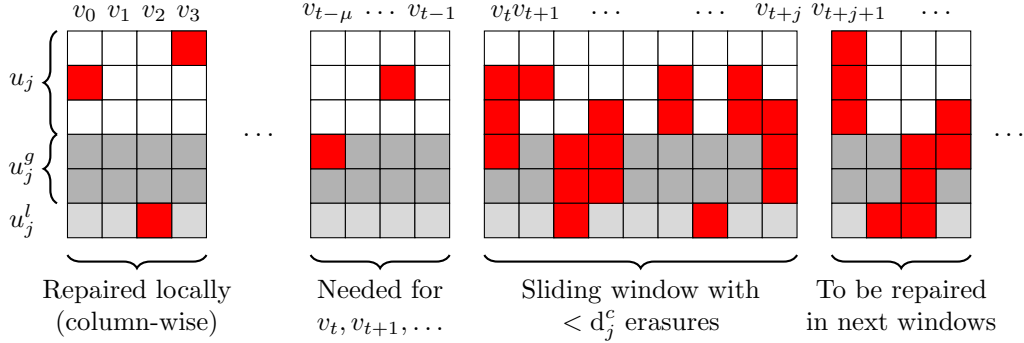
\begin{figure*} [!t]
\begin{center}
\begin{tabular}{c@{\extracolsep{1cm}}c}
\begin{tikzpicture}[
square/.style = {draw, rectangle, 
                 minimum size=\m, outer sep=0, inner sep=0, font=\small,
                 },
                        ]
\def\m{13pt}
\def\w{26} 
\def\wm{25}
\def\h{6}
\def\glob{3}
\def\loc{5}
    \pgfmathsetmacro\uw{int(\w/2)}
    \pgfmathsetmacro\uh{int(\h/2)}

\foreach \x in {0,...,3}
{
  \node [] at (\x*\m + \m, 0.01) {$ v_{\x} $};
}

  \foreach \x in {1,...,4}
  {
    \foreach \y in {1,...,\h}
       {    
           
           \ifnum\y>\loc
               \node [square, fill=gray!30]  (\x,\y) at (\x*\m,-\y*\m) {$  $};
           \else
               \ifnum\y>\glob
                   \node [square, fill=gray!60]  (\x,\y) at (\x*\m,-\y*\m) {$  $};
               \else
                   \node [square, fill=white]  (\x,\y) at (\x*\m,-\y*\m) {$  $};
               \fi
           \fi
       }
  }
  
\node [square, fill=red!100]  (1,2) at (1*\m,-2*\m) {$  $};  
\node [square, fill=red!100]  (3,6) at (3*\m,-6*\m) {$  $};
\node [square, fill=red!100]  (4,1) at (4*\m,-1*\m) {$  $};
  
\node [] (1, \h) at (5*\m + 1*\m, -3.5*\m) {$ \ldots $};

\node [] at (7*\m + \m, 0.01) {$ v_{t - \mu} $};
\node [] at (8.5*\m + \m, 0.01) {$ \ldots $};
\node [] at (10*\m + \m, 0.01) {$ v_{t - 1} $};

  \foreach \x in {8,...,11}
  {
    \foreach \y in {1,...,\h}
       {    
           
           \ifnum\y>\loc
               \node [square, fill=gray!30]  (\x,\y) at (\x*\m,-\y*\m) {$  $};
           \else
               \ifnum\y>\glob
                   \node [square, fill=gray!60]  (\x,\y) at (\x*\m,-\y*\m) {$  $};
               \else
                   \node [square, fill=white]  (\x,\y) at (\x*\m,-\y*\m) {$  $};
               \fi
           \fi
       }
  }

\node [square, fill=red!100]  (8,4) at (8*\m,-4*\m) {$  $};
\node [square, fill=red!100]  (10,2) at (10*\m,-2*\m) {$  $};

\node [] at (12*\m + \m, 0.01) {$ v_{t} $};
\node [] at (13*\m + \m, 0.01) {$ v_{t + 1} $};
\node [] at (15*\m + \m, 0.01) {$ \ldots $};
\node [] at (18*\m + \m, 0.01) {$ \ldots $};
\node [] at (20*\m + \m, 0.01) {$ v_{t + j} $};

  \foreach \x in {13,...,21}
  {
    \foreach \y in {1,...,\h}
       {    
           
           \ifnum\y>\loc
               \node [square, fill=gray!30]  (\x,\y) at (\x*\m,-\y*\m) {$  $};
           \else
               \ifnum\y>\glob
                   \node [square, fill=gray!60]  (\x,\y) at (\x*\m,-\y*\m) {$  $};
               \else
                   \node [square, fill=white]  (\x,\y) at (\x*\m,-\y*\m) {$  $};
               \fi
           \fi
       }
  }

\node [square, fill=red!100]  (13,2) at (13*\m,-2*\m) {$  $};   
\node [square, fill=red!100]  (13,3) at (13*\m,-3*\m) {$  $};   
\node [square, fill=red!100]  (13,4) at (13*\m,-4*\m) {$  $};  
\node [square, fill=red!100]  (14,2) at (14*\m,-2*\m) {$  $};   

\node [square, fill=red!100]  (15,4) at (15*\m,-4*\m) {$  $}; 
\node [square, fill=red!100]  (15,5) at (15*\m,-5*\m) {$  $}; 
\node [square, fill=red!100]  (15,6) at (15*\m,-6*\m) {$  $}; 

\node [square, fill=red!100]  (16,3) at (16*\m,-3*\m) {$  $}; 
\node [square, fill=red!100]  (16,4) at (16*\m,-4*\m) {$  $};   
\node [square, fill=red!100]  (16,5) at (16*\m,-5*\m) {$  $};
   
\node [square, fill=red!100]  (18,2) at (18*\m,-2*\m) {$  $};   
\node [square, fill=red!100]  (18,3) at (18*\m,-3*\m) {$  $}; 

\node [square, fill=red!100]  (19,6) at (19*\m,-6*\m) {$  $}; 

\node [square, fill=red!100]  (20,2) at (20*\m,-2*\m) {$  $};   
\node [square, fill=red!100]  (20,3) at (20*\m,-3*\m) {$  $};  
   
\node [square, fill=red!100]  (21,3) at (21*\m,-3*\m) {$  $}; 
\node [square, fill=red!100]  (21,4) at (21*\m,-4*\m) {$  $};   
\node [square, fill=red!100]  (21,5) at (21*\m,-5*\m) {$  $};

\node [] at (24.5*\m + \m, 0.01) {$ \ldots $};  
\node [] at (22*\m + \m, 0.01) {$ v_{t + j + 1} $};  

  \foreach \x in {23,...,26}
  {
    \foreach \y in {1,...,\h}
       {    
           
           \ifnum\y>\loc
               \node [square, fill=gray!30]  (\x,\y) at (\x*\m,-\y*\m) {$  $};
           \else
               \ifnum\y>\glob
                   \node [square, fill=gray!60]  (\x,\y) at (\x*\m,-\y*\m) {$  $};
               \else
                   \node [square, fill=white]  (\x,\y) at (\x*\m,-\y*\m) {$  $};
               \fi
           \fi
       }
  }  
\node [square, fill=red!100]  (23,1) at (23*\m,-1*\m) {$  $}; 
\node [square, fill=red!100]  (23,2) at (23*\m,-2*\m) {$  $};   
\node [square, fill=red!100]  (23,3) at (23*\m,-3*\m) {$  $};      
    
\node [square, fill=red!100]  (24,6) at (24*\m,-6*\m) {$  $}; 
   
\node [square, fill=red!100]  (25,4) at (25*\m,-4*\m) {$  $};   
\node [square, fill=red!100]  (25,5) at (25*\m,-5*\m) {$  $};    
\node [square, fill=red!100]  (25,6) at (25*\m,-6*\m) {$  $}; 

\node [square, fill=red!100]  (26,3) at (26*\m,-3*\m) {$  $};   
\node [square, fill=red!100]  (26,4) at (26*\m,-4*\m) {$  $};

\node [] (\w + 1, \h) at (\w*\m + \m + 0.6*\m, -3.5*\m) {$ \ldots $};

   \draw [decorate, thick,decoration={brace,amplitude=5pt}]
   (0.2*\m,-3.5*\m) -- (0.2*\m,-0.5*\m) node[midway,xshift=-1.3em]{$ u_j $};    
   \draw [decorate, thick,decoration={brace,amplitude=5pt}]
   (0.2*\m,-5.5*\m) -- (0.2*\m,-3.5*\m) node[midway,xshift=-1.3em]{$ u^g_j $};    
   \draw [draw=none, decorate, thick,decoration={brace,amplitude=5pt}]
   (0.2*\m,-6.5*\m) -- (0.2*\m,-5.5*\m) node[midway,xshift=-1.3em]{$ u^l_j $};

   \draw [draw, decorate, thick,decoration={brace,amplitude=5pt,mirror}]
   (0.5*\m,-7*\m) -- (4.5*\m,-7*\m) node[midway,yshift=-2em]{\begin{tabular}{c}Repaired locally\\ (column-wise) \end{tabular}};        
   \draw [draw, decorate, thick,decoration={brace,amplitude=5pt,mirror}]
   (7.5*\m,-7*\m) -- (11.5*\m,-7*\m) node[midway,yshift=-2em]{\begin{tabular}{c}Needed for\\ $ v_t, v_{t+1}, \ldots $ \end{tabular}};       
   \draw [draw, decorate, thick,decoration={brace,amplitude=5pt,mirror}]
   (12.5*\m,-7*\m) -- (21.5*\m,-7*\m) node[midway,yshift=-2em]{\begin{tabular}{c}Sliding window with\\ $ < \dd^c_j $ erasures \end{tabular}};        
   \draw [draw, decorate, thick,decoration={brace,amplitude=5pt,mirror}]
   (22.5*\m,-7*\m) -- (26.5*\m,-7*\m) node[midway,yshift=-2em]{\begin{tabular}{c}To be repaired\\ in next windows \end{tabular}};

\end{tikzpicture}

\end{tabular}
\end{center}

\caption{Sliding-window repair combined with local repair. Here, an $ (n,k,r,\partial) = (6,3,5,2) $ LRCC, with one local group ($ g = 1 $) per block of $ n $ symbols, is depicted. Each column depicts a systematic encoded block $ v_j = (u_j, u^g_j, u^l_j) $. White, light grey and dark grey boxes denote information symbols $ u_j $, global parities $ u^g_j $ and local parities $ u^l_j $, respectively. The local parities can be invoked block-wise to correct one erasure per block (since $ \partial = 2 $), requiring the other $ r=5 $ symbols in the block for repair. If catastrophic erasures occur, with $ < \dd^c_j $ erasures in each window $ (v_t, v_{t+1}, \ldots, v_{t+j}) $, then sliding-window repair is invoked, as described in Theorem \ref{th sliding window}. Observe that, although sliding-window repair is somehow local, by setting $ j=0 $, we see that we still need to contact and download the previous $ \mu n $ symbols from $ v_{t-\mu}, v_{t-\mu+1}, \ldots, v_{t-1} $ in order to repair one symbol in $ v_t $, hence cannot compete with the considered type of locality even for relatively small memory $ \mu $. }
\label{fig sliding window repair}
\end{figure*}

As it was the case for locally repairable block codes, the main goal, given the parameters $ n $, $ k $, $ r $ and $ \partial $ (and now also $ \delta $ and $ \mu $), is to obtain a corresponding LRCC with \textit{maximum global distance} properties, which would allow for global erasure correction in case of catastrophic failures. In this work, we consider column distances for ``global correction'', since we will focus on sliding-window erasure correction as in Theorem \ref{th sliding window}. See Fig. \ref{fig sliding window repair} for a graphical description of local repair combined with sliding-window global repair.

In the next theorem, we provide a Singleton bound on column distances of LRCCs. As the reader can see, we need to make three assumptions for the general bound, the first being that local groups are pair-wise disjoint and of full length, the second being that $ r $ divides $ k $, and the third is that a smallest possible subset of local groups form an information set of the code. This latter condition is satisfied if the $ 0 $th column distance is optimal, as stated in the theorem. 

\begin{theorem} \label{th singleton bound on column distances}
Let $ \mathcal{C} \subseteq \mathbb{F}[D]^n $ be an $ (n,k,r,\partial) $ LRCC with a reduced generator matrix $ G(D) = \sum_{h=0}^\mu G_h D^h $ such that $ G_0 $ is full-rank. Then it holds that
\begin{equation}
\dd_0^c(\mathcal{C}) \leq (n-k) - \left(  \left\lceil \frac{k}{r} \right\rceil  - 1 \right)(\partial - 1) + 1 .
\label{eq singleton bound with localities case j=0}
\end{equation}
Now assume that $ k = \ell r $, for a positive integer $ \ell $, local groups are pair-wise disjoint (i.e., $ \Gamma_i \cap \Gamma_j = \varnothing $ if $ i \neq j $) and of full size $ r + \partial - 1 $, and that there exist $ \ell $ of them forming an information set for the $ k $-dimensional linear block code $ \mathcal{C}_0^c \cup \{ 0 \} \subseteq \mathbb{F}^n $, generated by $ G_0 $. This latter condition holds if equality is achieved in (\ref{eq singleton bound with localities case j=0}), by Lemma \ref{lemma info sets optimal block lrc} in Appendix \ref{app lemma optimal block lrc}. Then it holds that
\begin{equation}
\dd_j^c(\mathcal{C}) \leq (n-k)(j+1) - \left( \frac{k(j+1)}{r}  - 1 \right)(\partial - 1) + 1,
\label{eq singleton bound with localities}
\end{equation}
for all $ j \in \mathbb{N} $.
\end{theorem}
\begin{proof}
We start by observing that $ G_0^c = G_0 $ and the block code 
\begin{equation}
\mathcal{C}_0 = \{ u G_0 \in \mathbb{F}^n \mid u \in \mathbb{F}^k \} = \mathcal{C}_0^c \cup \{ 0 \} \subseteq \mathbb{F}^n
\label{eq def block code for proof singleton bound}
\end{equation}
is a $ k $-dimensional linear block LRC of length $ n $ with $ (r , \partial) $-localities. Hence the bound (\ref{eq singleton bound with localities case j=0}) is the classical upper bound on the minimum Hamming distance of linear block LRCs \cite[Th. 2.1]{kamath}. 

We will now prove the bound (\ref{eq singleton bound with localities}), for $ j \in \mathbb{N} $, given the assumptions in the theorem. Assume that local groups are pair-wise disjoint, and that $ |\Gamma_i| = r + \partial - 1 $, for $ i = 1,2, \ldots, g $. Finally, assume, without loss of generality, that the first $ \ell $ local groups $ \Gamma_1, \Gamma_2, \ldots, \Gamma_\ell $ form an information set for the linear block code $ \mathcal{C}_0 = \mathcal{C}_0^c \cup \{ 0 \} $. 

Let $ \Delta_i \subseteq \Gamma_i $ denote the first $ r $ coordinates in $ \Gamma_i $, for $ i = 1,2, \ldots, g $. By the $ (r,\partial) $-locality of $ \mathcal{C}_0 $, the set 
$$ \Delta = \Delta_1 \cup \Delta_2 \cup \ldots \cup \Delta_\ell \subseteq \Gamma_1 \cup \Gamma_2 \cup \ldots \cup \Gamma_\ell $$
is an information set of $ \mathcal{C}_0 $ of size $ k = \ell r $. Hence we may perform row operations on the generator matrix $ G_0 \in \mathbb{F}^{k \times n} $ of $ \mathcal{C}_0 $ to obtain a systematic generator matrix of the form
\begin{equation}
G_0^\prime = \left( I_{r,1} , A_{0,1} | I_{r,2} , A_{0,2} | \ldots | I_{r,\ell} , A_{0,\ell} | B_{0,1} | B_{0,2} | \ldots | B_{0,g - \ell}  \right) \in \mathbb{F}^{k \times n} , 
\label{eq def G_0^prime}
\end{equation}
for matrices $ A_{0,1}, A_{0,2}, \ldots, $ $ A_{0,\ell} \in \mathbb{F}^{k \times (\partial - 1)} $ and $ B_{0,1}, B_{0,2}, \ldots, $ $ B_{0,g - \ell} \in \mathbb{F}^{k \times (r + \partial - 1)} $, and where $ I_{r,1}, I_{r,2}, \ldots, $ $ I_{r,\ell} \in \mathbb{F}^{k \times r} $ are such that
$$ I_k = \left( I_{r,1}, I_{r,2}, \ldots, I_{r,\ell} \right) \in \mathbb{F}^{k \times k} $$
is the $ k \times k $ identity matrix.

Fix now $ j \in \mathbb{N} $. Using the systematic generator matrix of $ \mathcal{C}_0 $ from (\ref{eq def G_0^prime}), we may perform row operations on the $ j $th \textit{truncated sliding generator matrix} $ G_j^c \in \mathbb{F}^{(j+1)k \times (j+1)n} $ from Definition \ref{def column distance} to obtain a row equivalent matrix (i.e., a matrix with the same row space) of the form
\begin{equation}
 \widetilde{G}_j^c = \left[ \begin{array}{cccc}
G_0^\prime & G_1^\prime & \ldots & G_j^\prime \\
 & G_0 & \ldots & G_{j-1} \\
 & & \ddots & \vdots \\
 & & & G_0
\end{array} \right] \in \mathbb{F}^{(j+1)k \times (j+1)n},
\label{eq def truncated matrix system for singleton proof}
\end{equation}
such that
\begin{equation}
G_h^\prime = \left( 0_{k,r} , A_{h,1} | 0_{k,r} , A_{h,2} | \ldots | 0_{k,r} , A_{h, \ell} | B_{h,1} | B_{h,2} | \ldots | B_{h, g - \ell}  \right) \in \mathbb{F}^{k \times n} , 
\label{eq def G_h^prime}
\end{equation}
for matrices $ A_{h,1}, A_{h,2}, \ldots, $ $ A_{h,\ell} \in \mathbb{F}^{k \times (\partial - 1)} $ and $ B_{h,1}, B_{h,2}, \ldots, $ $ B_{h,g - \ell} \in \mathbb{F}^{k \times (r + \partial - 1)} $, for $ h = 1,2, \ldots, j $, and where $ 0_{k,r} \in \mathbb{F}^{k \times r} $ denotes the $ k \times r $ zero matrix.

Now, let $ v = (v_0, v_1, \ldots, v_j) \in \mathbb{F}^{(j+1)n} $ be the first row of the matrix $ \widetilde{G}_j^c \in \mathbb{F}^{(j+1)k \times (j+1)n} $ from (\ref{eq def truncated matrix system for singleton proof}). By (\ref{eq def G_h^prime}), we have that
\begin{equation*}
\begin{split}
v_0 & = \left( 1, 0_{r-1}, a_{0,1} | 0_r , a_{0,2} | \ldots | 0_r , a_{0,\ell} | b_{0,1} | b_{0,2} | \ldots | b_{0,g - \ell} \right) , \\
v_1 & = \left( 0_r, a_{1,1} | 0_r , a_{1,2} | \ldots | 0_r, a_{1, \ell} | b_{1,1} | b_{1,2} | \ldots | b_{1, g-\ell} \right) , \\
\vdots &  \\
v_j & = \left( 0_r, a_{j,1} | 0_r , a_{j,2} | \ldots | 0_r, a_{j, \ell} | b_{j,1} | b_{j,2} | \ldots | b_{j, g-\ell} \right),
\end{split}
\end{equation*}
for vectors $ a_{h,1}, a_{h,2}, \ldots, $ $ a_{h, \ell} \in \mathbb{F}^{\partial - 1} $, $ b_{h,1}, b_{h,2} , \ldots, b_{h, g-\ell} \in \mathbb{F}^{r + \partial - 1} $, and where $ 0_r \in \mathbb{F}^r $ denotes the zero vector of length $ r $.

Clearly $ v \in \mathcal{C}_j^c $, since it is a linear combination of rows of the matrix $ G_j^c $ from Definition \ref{def column distance}, and its first block of $ n $ components is nonzero, that is, $ v_0 \neq 0 $.

Finally, since $ \mathcal{C} $ is an LRCC, by Item 2 in Definition \ref{def lrc convolutional codes}, we deduce that
\begin{equation*}
\begin{split}
 a_{0,2} = a_{0,3}  = \ldots = a_{0,\ell} & = 0_{\partial - 1} , \\
 a_{1,1} = a_{1,2} = a_{1,3}  = \ldots = a_{1,\ell} & = 0_{\partial - 1} , \\
\vdots &  \\
 a_{j,1} = a_{j,2} = a_{j,3}  = \ldots = a_{j,\ell} & = 0_{\partial - 1} , \\
\end{split}
\end{equation*}
where $ 0_{\partial - 1} \in \mathbb{F}^{\partial - 1} $ is the zero vector of length $ \partial - 1 $. 

Therefore, we conclude that
\begin{equation*}
\begin{split}
\wt_H(v) & \leq (j+1)n - \ell (j+1) (r + \partial - 1) + \partial \\
 & = (j+1)n - (j+1) \ell r - ((j+1) \ell - 1)(\partial - 1) + 1 \\
 & = (n-k) (j+1) - \left( \frac{k(j+1)}{r} - 1 \right) (\partial - 1) + 1,
\end{split}
\end{equation*}
and we are done.
\end{proof}

In the next section, we show how to explicitly construct a non-catastrophic LRCC attaining the previous bound, for all $ j= 0,1,2, \ldots, L $, where $ L $ is as in (\ref{eq def of L}), for fields of any characteristic but sufficiently large.

\begin{remark}
Recall that, by Proposition \ref{prop monotonicity singleton column distance}, a convolutional code that is $ j $-MDS is also $ h $-MDS, for all $ h = 0,1,2, \ldots, j $. However, it is not clear whether a code attaining the bound (\ref{eq singleton bound with localities}) for some $ j $ implies attaining the bound for $ h < j $. We leave this as an open problem. In any case, Construction \ref{construction} below based on a $ j $-MSRD convolutional code attains the bound (\ref{eq singleton bound with localities}) for all $ h = 0,1,2, \ldots, j $.
\end{remark}

\section{LRCCs based on Sum-Rank Convolutional Codes} \label{sec connecting LRCC and SR convolutional}

In this section, we will show how to construct non-catastrophic LRCCs attaining the bound in Theorem \ref{th singleton bound on column distances}, for $ j = 0,1,\ldots, L $, using a \textit{$ j $-MSRD convolutional code} (see Definition \ref{def mu-MSRD convolutional} below). To that end, we will use the notion of \textit{sum-rank weight} on each block of a convolutional code. Sum-rank weights were first defined in \cite{multishot} for error correction in multishot network coding (see also \cite{mahmood-convolutional, linearizedRS, secure-multishot, mrd-convolutional, wachter-convolutional} and the references therein). They were implicitly considered earlier in the space-time coding literature (see \cite[Sec. III]{space-time-kumar}), and they have been first used for locally repairable block codes in \cite{lrc-sum-rank}.

Throughout this section, we will fix a prime power $ q $ and a positive integer $ m $, and we will assume that $ \mathbb{F} = \mathbb{F}_{q^m} $. Fix an ordered basis $ \mathcal{A} = \{ \alpha_1, \alpha_2, \ldots, \alpha_m \} $ of $ \mathbb{F}_{q^m} $ over $ \mathbb{F}_q $. For any positive integer $ s $, we denote by $ M_\mathcal{A} :
\mathbb{F}_{q^m}^s \longrightarrow \mathbb{F}_q^{m \times s} $ the corresponding \textit{matrix representation} map, given by
\begin{equation}
M_\mathcal{A} \left(  c  \right) = \left[ \begin{array}{cccc}
c_{11} & c_{12} & \ldots & c_{1s} \\
c_{21} & c_{22} & \ldots & c_{2s} \\
\vdots & \vdots & \ddots & \vdots \\
c_{m1} & c_{m2} & \ldots & c_{ms} \\
\end{array} \right] \in \mathbb{F}_q^{m \times s},
\label{eq def matrix representation map}
\end{equation}
where $ c = \sum_{i=1}^m \alpha_i (c_{i,1}, c_{i,2}, \ldots, c_{i,s}) \in \mathbb{F}_{q^m}^s $ and $ c_{i,j} \in \mathbb{F}_q $, for $ i = 1,2, \ldots, m $ and $ j = 1,2, \ldots, s $.

Throughout this section, we will also fix a number of local groups $ g $, a locality $ r $, and the \textit{sum-rank length decomposition} $ N = gr $. The following definition is given in \cite{multishot}.

\begin{definition} [\textbf{\cite{multishot}}]
Let $ c = (c^{(1)}, $ $ c^{(2)}, $ $ \ldots,
$ $ c^{(g)}) \in \mathbb{F}_{q^m}^N $, where $
c^{(i)} \in \mathbb{F}_{q^m}^r $, for $ i = 1,2, \ldots,
g $. We define the \textit{sum-rank weight} of $ c $ as
$$ \wt_{SR}(c) = \sum_{i=1}^g {\rm
Rk}(M_{\mathcal{A}}(c^{(i)})). $$
\end{definition}

We extend sum-rank weights to convolutional codes as follows.

\begin{definition}
Given $ v(D) = \sum_{j \in \mathbb{N}} v_j D^j \in \mathbb{F}_{q^m}[D]^N $, we define its \textit{sum-rank weight} as
$$ \wt_{SR}(v(D)) = \sum_{j \in \mathbb{N}} \wt_{SR}(v_j). $$

Given an $ (N,k) $ convolutional code $ \mathcal{C} \subseteq \mathbb{F}_{q^m}[D]^N $, we define its \textit{sum-rank free distance} as
\begin{equation*}
\dd_{SR}(\mathcal{C}) = \min \{ \wt_{SR}(v(D)) \mid v(D) \in \mathcal{C} \textrm{ and } v(D) \neq 0 \} .
\end{equation*}
Finally, we define the \textit{$ j $th sum-rank column distance} of $ \mathcal{C} $ as
$$ \dd_{SR, j}^c(\mathcal{C}) =  \min \left\lbrace {\rm wt}_{SR}(v) \mid v \in \mathcal{C}_j^c \right\rbrace , $$
where $ \mathcal{C}_j^c $ is as in Definition \ref{def column distance}, in particular $ v \neq 0 $ if $ v \in \mathcal{C}_j^c $, for $ j \in \mathbb{N} $.
\end{definition}

Observe that, for any $ c = (c^{(1)}, c^{(2)}, \ldots, c^{(g)}) \in \mathbb{F}_{q^m}^N $, where $
c^{(i)} \in \mathbb{F}_{q^m}^r $, for $ i = 1,2, \ldots, g $, it holds that
$$ \wt_{SR}(c) = \sum_{i=1}^{g} \Rk(M_\mathcal{A}(c^{(i)})) \leq \sum_{i=1}^{g} \wt(c^{(i)}) = \wt(c), $$
since the rank of a matrix is at most the number of its non-zero columns. Hence, the following result follows immediately from its Hamming-metric counterpart (Propositions \ref{prop singleton column distance} and \ref{prop monotonicity singleton column distance}).

\begin{proposition} \label{prop singleton bound on column SR distances}
Given a non-catastrophic $ (N,k) $ convolutional code $ \mathcal{C} \subseteq \mathbb{F}_{q^m}[D]^N $, it holds that
\begin{equation}
\dd_{SR, j}^c(\mathcal{C}) \leq (N-k)(j+1) + 1,
\label{eq singleton bound on column SR distances}
\end{equation}
for all $ j \in \mathbb{N} $. Furthermore, the following hold.
\begin{enumerate}
\item
If $ \dd_{SR, j+1}^c(\mathcal{C}) = (N-k)(j+2) + 1 $, then $ \dd_{SR, j}^c(\mathcal{C}) = (N-k)(j+1) + 1 $, for $ j \in \mathbb{N} $.
\item
If $ \dd_{SR, j}^c(\mathcal{C}) = (N-k)(j+1) + 1 $, then $ j \leq L $, where $ L $ is as in (\ref{eq def of L}).
\end{enumerate}
\end{proposition}

The previous proposition motivates the following definition.

\begin{definition} \label{def mu-MSRD convolutional}
We say that an $ (N,k) $ convolutional code $ \mathcal{C} \subseteq \mathbb{F}_{q^m}[D]^N $ is $ j $-maximum-sum-rank-distance, or $ j $-MSRD for short, if it is non-catastrophic and $ \dd_{SR, j}^c(\mathcal{C}) = (N-k)(j+1) + 1 $. 
\end{definition}

We now describe how to construct LRCCs from sum-rank convolutional codes. This construction is inspired by \cite[Const. I]{rawat}.

\begin{construction} \label{construction}
Assume that $ q \geq r + \partial - 1 $, and choose:
\begin{enumerate}
\item
\textit{Outer code}: An $ (N,k) $ convolutional code $ \mathcal{C}_{out} \subseteq \mathbb{F}_{q^m}[D]^N $.
\item
\textit{Local codes}: An MDS $ (r +\partial -1, r) $ block code $ \mathcal{C}_{loc} \subseteq \mathbb{F}_q^{r + \partial - 1} $ with generator matrix $ A \in \mathbb{F}_q^{r \times (r + \partial - 1)} $.
\item
\textit{Global code}: We define the global code $ \mathcal{C}_{glob} \subseteq \mathbb{F}_{q^m}[D]^n $, with $ n = (r + \partial - 1)g = N + (\partial - 1)g $, as the $ (n,k) $ convolutional code given by 
$$ \mathcal{C}_{glob} = \left\lbrace \sum_{j \in \mathbb{N}} (v_j \diag_g(A)) D^j \mid \sum_{j \in \mathbb{N}} v_j D^j \in \mathcal{C}_{out} \right\rbrace , $$
where $ \diag_g(A) $ is defined as a block-diagonal matrix with $ A \in \mathbb{F}_q^{r \times (r + \partial - 1)} $ repeated $ g $ times (recall that $ N = gr $ and $ n = g(r+\partial-1) $):
$$ \diag_g(A) = \diag(A, A, \ldots, A) = \left[ \begin{array}{cccc}
A & 0 & \ldots & 0 \\
0 & A & \ldots & 0 \\
\vdots & \vdots & \ddots & \vdots \\
0 & 0 & \ldots & A
\end{array} \right] \in \mathbb{F}_q^{N \times n} . $$
\end{enumerate}
\end{construction}

Observe that if $ G_{out}(D) = \sum_{j=0}^\mu G_{out, j} D^j \in \mathbb{F}_{q^m}[D]^{k \times N} $ is a generator matrix of $ \mathcal{C}_{out} \subseteq \mathbb{F}_{q^m}[D]^N $, then a generator matrix of $ \mathcal{C}_{glob} \subseteq \mathbb{F}_{q^m}[D]^n $ is simply given by
$$ G_{glob}(D) = \sum_{j=0}^\mu G_{glob, j} D^j = \sum_{j=0}^\mu \left( G_{out, j} \diag_g(A) \right) D^j \in \mathbb{F}_{q^m}[D]^{k \times n}. $$
In addition, note that multiplying a vector $ v(D) \in \mathbb{F}[D]^N $ on the right by a rank-$ N $ constant matrix $ C \in \mathbb{F}^{N \times n} $ preserves the degree of $ v(D) $. Hence if $ G_{out}(D) $ is reduced, then so is $ G_{glob}(D) $. It also follows easily that if $ G_{out}(D) $ is basic, then so is $ G_{glob}(D) $. Thus we deduce the following.

\begin{lemma} \label{lemma restricted const 1 is non-catastrophic}
In Construction \ref{construction}, if $ \mathcal{C}_{out} $ is non-catastrophic, then so are $ \mathcal{C}_{glob} $ and $ (\mathcal{C}_{glob})_\Delta $, for any $ \Delta = \bigcup_{i=1}^g \Delta_i \subseteq [n] $ such that $ \Delta_i \subseteq \Gamma_i $ and $ |\Delta_i| \geq r $, for $ i = 1,2, \ldots,g $. Here, we denote $ \Gamma_i = [(r+\partial -1)(i-1) + 1, (r+\partial -1)i] \subseteq [n] $, for $ i = 1,2, \ldots, g $.
\end{lemma}

As it was the case for locally repairable block codes (see \cite[Lemma 1]{lrc-sum-rank}), any LRCC whose local codes are all encoded by the same linear MDS code over the subfield $ \mathbb{F}_q $, are necessarily of the form of Construction \ref{construction}. For this reason, Construction \ref{construction} not only is natural, but it is somehow unavoidable. 

We may now prove the main result of this section, which states that $ \mathcal{C}_{glob} $ in Construction \ref{construction} has maximum $ h $th sum-rank column distance among all non-catastrophic $ (n,k,r,\partial) $ LRCC, for $ h = 0,1,2, \ldots, j $, if $ \mathcal{C}_{out} $ is $ j $-MSRD.

\begin{theorem} \label{th MSRD implies optimal LRCC}
In Construction \ref{construction}, $ \mathcal{C}_{glob} $ is an $ (n,k,r,\partial) $ LRCC. Furthermore, if $ j \in \mathbb{N} $ and $ \mathcal{C}_{out} $ is $ j $-MSRD (thus non-catastrophic), then $ \mathcal{C}_{glob} $ is non-catastrophic and  
\begin{equation}
\dd_h^c (\mathcal{C}_{glob}) = (n-k)(h+1) - \left(  \left\lceil \frac{k(h+1)}{r} \right\rceil  - 1 \right)(\partial - 1) + 1,
\label{eq d_ell attains local singleton}
\end{equation}
for all $ h = 0,1,2, \ldots, j $.
\end{theorem}
\begin{proof}
First, it follows easily from the definitions and Construction \ref{construction} that $ \mathcal{C}_{glob} $ is an $ (n,k,r,\partial) $ LRCC. The non-catastrophic property is part of Lemma \ref{lemma restricted const 1 is non-catastrophic}. Therefore, we only need to show that (\ref{eq d_ell attains local singleton}) holds for $ h = j $, since if $ \mathcal{C}_{out} $ is $ j $-MSRD, then $ \mathcal{C}_{out} $ is $ h $-MSRD, for all $ h = 0,1,2, \ldots, j $, by Proposition \ref{prop singleton bound on column SR distances}.

Now we need to show that, for any $ v \in (\mathcal{C}_{glob})_j^c \subseteq \mathbb{F}_{q^m}^{n(j+1)} $, the non-zero coordinates of $ v $ are not all inside some pattern of
$$ e = (n-k)(j+1) - \left(  \left\lceil \frac{k(j+1)}{r} \right\rceil  - 1 \right)(\partial - 1)  $$
erasures in the block $ [n(j+1)] $ of coordinates. 

Assume the opposite holds, that is, there exists $ v \in (\mathcal{C}_{glob})_j^c $ with all of its non-zero coordinates in an erasure pattern of size $ e $. Observe that $  v \in \mathbb{F}_{q^m}^{n(j+1)} $ is a block codeword, and by construction, there exists $ x \in (\mathcal{C}_{out})_j^c $ such that $ v = x \diag_{g(j+1)}(A) $.

Let $ \mathcal{E}_{gh + i} \subseteq [r+\partial-1] $ be the erasure pattern in the $ i $th local group in the $ h $th block of $ n $ coordinates, and define $ \mathcal{R}_{gh + i} = [r+\partial-1] \setminus \mathcal{E}_{gh + i} $, for $ i = 1,2, \ldots, g $ and $ h = 0,1,2, \ldots, j $. The truncated global codeword after removing all the symbols in such an erasure pattern is by assumption the zero vector, that is,
\begin{equation}
0 = x \diag(A|_{\mathcal{R}_1}, A|_{\mathcal{R}_2}, \ldots, A|_{\mathcal{R}_{g(j+1)}}) \in \mathbb{F}_{q^m}^{n(j+1) - e},
\label{eq codeword with erasure pattern}
\end{equation}
where $ e = \sum_{u = 1}^{g(j+1)} |\mathcal{E}_u| = n(j+1) - \sum_{u = 1}^{g(j+1)} |\mathcal{R}_u| $.

Assume for simplicity that $ k(j+1) = \ell r $, for some integer $ \ell \in \mathbb{N} $. Note that we may decompose
\begin{equation}
e = n(j+1) - (r + \partial - 1)\ell + \partial -1.
\label{eq decomposition e for proof j-MDS}
\end{equation}
As discussed in the proof of \cite[Th. 24]{rawat}, the best-case erasure pattern is obtained when erasures concentrate in the smallest number of local groups. Here, by best-case erasure pattern we mean an erasure pattern whose complement set of coordinates contain the most locally redundant symbols, which means that $ \sum_{u=1}^{g(j+1)} \Rk(A|_{\mathcal{R}_u}) $ is the minimum possible. Thus by (\ref{eq decomposition e for proof j-MDS}), in the best case we have without loss of generality that $ \mathcal{R}_u = [r+\partial-1] $, for $ u = 1,2, \ldots, \ell - 1 $, $ |\mathcal{R}_\ell| = r $, and $ \mathcal{R}_u = \varnothing $, for $ u = \ell+1 , \ell+2 , \ldots, g(j+1) $. Since $ \mathcal{C}_{loc} $ is an $ (r+\partial-1, r) $ MDS code, we have that $ \Rk(A|_{\mathcal{R}_u}) = r $, for $ u = 1,2, \ldots, \ell $. Therefore, in the best case, we have that
$$ \sum_{u=1}^{g(j+1)} \Rk(A|_{\mathcal{R}_u}) = \ell r = k(j+1). $$
Define now $ \mathcal{R}^\prime_u \subseteq [r+\partial-1] $ as the set formed by some $ r $ coordinates in $ \mathcal{R}_u \subseteq [r+\partial-1] $, for $ u = 1,2, \ldots, \ell $. Define also $ \mathcal{R}^\prime_u \subseteq [r+\partial - 1] $ as any $ r $ coordinates in $ [r+\partial-1] $, for $ u = \ell+1, \ell+2, \ldots, g(j+1) $. Since $ \mathcal{C}_{loc} $ is an $ (r+\partial-1, r) $ MDS code, we have that $ {\rm Rk}(A|_{\mathcal{R}^\prime_u}) = r $, that is, $ A|_{\mathcal{R}^\prime_u} \in \mathbb{F}_q^{r \times r} $ is invertible, for $ u = 1,2, \ldots, g(j+1) $. Therefore, we conclude that
\begin{equation*}
\begin{split}
{\rm wt}_{SR} (x) & = {\rm wt}_{SR} (x \diag(A|_{\mathcal{R}^\prime_1}, A|_{\mathcal{R}^\prime_2}, \ldots, A|_{\mathcal{R}^\prime_{g(j+1)}})) \\
 & \leq {\rm wt}_H (x \diag(A|_{\mathcal{R}^\prime_1}, A|_{\mathcal{R}^\prime_2}, \ldots, A|_{\mathcal{R}^\prime_{g(j+1)}})) \\
 & \leq g(j+1)r - \ell r = (gr - k)(j+1) = (N-k)(j+1),
\end{split}
\end{equation*}
where the last inequality follows from (\ref{eq codeword with erasure pattern}). This is absurd since $ x \in (\mathcal{C}_{out})_j^c $ and
$$ \dd_{SR, j}^c(\mathcal{C}_{out}) = (N-k)(j+1) + 1 . $$
We conclude that there is no $ v \in (\mathcal{C}_{glob})_j^c $ whose non-zero coordinates are all inside some pattern of $ e $ erasures, hence $ \dd_{SR, j}^c(\mathcal{C}_{glob}) \geq e+1 $, and we are done.
\end{proof}

We conclude by plugging in Construction \ref{construction} the MSRD convolutional codes from \cite{mahmood-convolutional} (see Appendix \ref{app MSRD conv codes}) as outer codes, and applying the previous theorem.

\begin{corollary} \label{cor j-MSRD applied to const 1}
If $ N = gr $, $ (N-k) | \delta $, $ M = \max \{ N-k,k \} $, $ L = \lfloor \frac{\delta}{k} \rfloor + \delta/(N-k) $, $ q \geq r+\partial - 1 $ and $ m \geq q^{M(L+2) - 1} $, then there exists a non-catastrophic $ (n,k,r,\partial) $ LRCC $ \mathcal{C}_{glob} \subseteq \mathbb{F}_{q^m}[D]^n $, of degree $ \delta $, satisfying (\ref{eq singleton bound with localities}) with equality, for $ j = 0,1,2, \ldots, L $, given as in Construction \ref{construction}, and where $ \mathcal{C}_{out} \subseteq \mathbb{F}_{q^m}[D]^N $ is the non-catastrophic $ L $-MSRD convolutional code in Appendix \ref{app MSRD conv codes}.
\end{corollary}

Corollary \ref{cor j-MSRD applied to const 1} not only shows that the upper bound given in (\ref{eq singleton bound with localities}) is sharp, but also provides an explicit class of codes that achieves such a bound. Moreover, these codes exist for any characteristic (in particular, when $ 2 | q $), and the local code may be arbitrary and with local fields of size $ q \approx r+\partial-1 $, which are small. We may also choose $ q = 2 $ if $ \partial = 2 $ and local repair would simply consist in XORing. Their main disadvantage is the huge exponent $ m $, which is in turn exponential in the degree $ \delta $ and in $ \max \{ N-k, k \} $. However, the bound on $ m $ in the corollary is only a bound, and there are cases when $ m $ can be chosen much smaller (see Table I in \cite{mahmood-convolutional}).

\section{Partial $j$-MDS and Partial MDP Convolutional Codes} \label{sec partial MDS convolutional}

In this section, we introduce partial MDP convolutional codes, analogous to the concept of \textit{partial MDS codes}, or LRC with \textit{maximal recoverability} (MR), introduced in \cite{blaum-RAID, gopalan-MR}. We will conclude by showing that the codes in Corollary \ref{cor j-MSRD applied to const 1} are partial MDP.

\begin{definition} \label{def partial MDS}
With notation as in Definition \ref{def lrc convolutional codes}, and for $ j \in \mathbb{N} $, we say that an $ (n,k,r,\partial) $ LRCC $ \mathcal{C} \subseteq \mathbb{F}[D]^n $ is \textit{partial $ j $-MDS} if the following holds: For all $ \Delta_i \subseteq \Gamma_i $ such that $ |\Gamma_i \setminus \Delta_i| = \partial - 1 $, for $ i = 1,2, \ldots, g $, the restricted $ (N,k) $ convolutional code $ \mathcal{C}_{\Delta} \subseteq \mathbb{F}[D]^N $ is non-catastrophic and $ j $-MDS (Definition \ref{def mds and mdp}), where $ \Delta = \bigcup_{i=1}^g \Delta_i $ and $ N = |\Delta| $.
\end{definition}

\begin{figure*} [!t]
\begin{center}

\begin{tabular}{c@{\extracolsep{1cm}}c}
\begin{tikzpicture}[
square/.style = {draw, rectangle, 
                 minimum size=\m, outer sep=0, inner sep=0, font=\small,
                 },
                        ]
\def\m{13pt}
\def\w{26} 
\def\wm{25}
\def\h{6}
\def\glob{3}
\def\loc{5}
    \pgfmathsetmacro\uw{int(\w/2)}
    \pgfmathsetmacro\uh{int(\h/2)}

\foreach \x in {0,...,3}
{
  \node [] at (\x*\m + \m, 0.01) {$ v_{\x} $};
}

  \foreach \x in {1,...,4}
  {
    \foreach \y in {1,...,\h}
       {    
           
           \ifnum\y>\loc
               \node [square, fill=gray!30]  (\x,\y) at (\x*\m,-\y*\m) {$  $};
           \else
               \ifnum\y>\glob
                   \node [square, fill=gray!60]  (\x,\y) at (\x*\m,-\y*\m) {$  $};
               \else
                   \node [square, fill=white]  (\x,\y) at (\x*\m,-\y*\m) {$  $};
               \fi
           \fi
       }
  }
  
\node [square, fill=red!100]  (1,2) at (1*\m,-2*\m) {$  $};  
\node [square, fill=red!100]  (3,6) at (3*\m,-6*\m) {$  $};
\node [square, fill=red!100]  (4,1) at (4*\m,-1*\m) {$  $};
  
\node [] (1, \h) at (5*\m + 1*\m, -3.5*\m) {$ \ldots $};

\node [] at (7*\m + \m, 0.01) {$ v_{t - \mu} $};
\node [] at (8.5*\m + \m, 0.01) {$ \ldots $};
\node [] at (10*\m + \m, 0.01) {$ v_{t - 1} $};

  \foreach \x in {8,...,11}
  {
    \foreach \y in {1,...,\h}
       {    
           
           \ifnum\y>\loc
               \node [square, fill=gray!30]  (\x,\y) at (\x*\m,-\y*\m) {$  $};
           \else
               \ifnum\y>\glob
                   \node [square, fill=gray!60]  (\x,\y) at (\x*\m,-\y*\m) {$  $};
               \else
                   \node [square, fill=white]  (\x,\y) at (\x*\m,-\y*\m) {$  $};
               \fi
           \fi
       }
  }

\node [square, fill=red!100]  (8,4) at (8*\m,-4*\m) {$  $};
\node [square, fill=red!100]  (10,2) at (10*\m,-2*\m) {$  $};

\node [] at (12*\m + \m, 0.01) {$ v_{t} $};
\node [] at (13*\m + \m, 0.01) {$ v_{t + 1} $};
\node [] at (15*\m + \m, 0.01) {$ \ldots $};
\node [] at (18*\m + \m, 0.01) {$ \ldots $};
\node [] at (20*\m + \m, 0.01) {$ v_{t + j} $};

  \foreach \x in {13,...,21}
  {
    \foreach \y in {1,...,\h}
       {    
           
           \ifnum\y>\loc
               \node [square, fill=gray!30]  (\x,\y) at (\x*\m,-\y*\m) {$  $};
           \else
               \ifnum\y>\glob
                   \node [square, fill=gray!60]  (\x,\y) at (\x*\m,-\y*\m) {$  $};
               \else
                   \node [square, fill=white]  (\x,\y) at (\x*\m,-\y*\m) {$  $};
               \fi
           \fi
       }
  }

\node [square, fill=red!100]  (13,2) at (13*\m,-2*\m) {$  $};   
\node [square, fill=red!100]  (13,3) at (13*\m,-3*\m) {$  $};   
\node [square, fill=red!100]  (13,4) at (13*\m,-4*\m) {$  $};  
\node [square, fill=red!100]  (14,2) at (14*\m,-2*\m) {$  $};   

\node [square, fill=red!100]  (15,4) at (15*\m,-4*\m) {$  $}; 
\node [square, fill=red!100]  (15,5) at (15*\m,-5*\m) {$  $}; 
\node [square, fill=red!100]  (15,6) at (15*\m,-6*\m) {$  $}; 

\node [square, fill=red!100]  (16,3) at (16*\m,-3*\m) {$  $}; 
\node [square, fill=red!100]  (16,4) at (16*\m,-4*\m) {$  $};   
\node [square, fill=red!100]  (16,5) at (16*\m,-5*\m) {$  $};
   
\node [square, fill=red!100]  (18,2) at (18*\m,-2*\m) {$  $};   
\node [square, fill=red!100]  (18,3) at (18*\m,-3*\m) {$  $}; 

\node [square, fill=red!100]  (19,6) at (19*\m,-6*\m) {$  $}; 

\node [square, fill=red!100]  (20,2) at (20*\m,-2*\m) {$  $};   
\node [square, fill=red!100]  (20,3) at (20*\m,-3*\m) {$  $};  
   
\node [square, fill=red!100]  (21,3) at (21*\m,-3*\m) {$  $}; 
\node [square, fill=red!100]  (21,4) at (21*\m,-4*\m) {$  $};   
\node [square, fill=red!100]  (21,5) at (21*\m,-5*\m) {$  $};

\node [] at (24.5*\m + \m, 0.01) {$ \ldots $};  
\node [] at (22*\m + \m, 0.01) {$ v_{t + j + 1} $};  

  \foreach \x in {23,...,26}
  {
    \foreach \y in {1,...,\h}
       {    
           
           \ifnum\y>\loc
               \node [square, fill=gray!30]  (\x,\y) at (\x*\m,-\y*\m) {$  $};
           \else
               \ifnum\y>\glob
                   \node [square, fill=gray!60]  (\x,\y) at (\x*\m,-\y*\m) {$  $};
               \else
                   \node [square, fill=white]  (\x,\y) at (\x*\m,-\y*\m) {$  $};
               \fi
           \fi
       }
  }  
\node [square, fill=red!100]  (23,1) at (23*\m,-1*\m) {$  $}; 
\node [square, fill=red!100]  (23,2) at (23*\m,-2*\m) {$  $};   
\node [square, fill=red!100]  (23,3) at (23*\m,-3*\m) {$  $};      
    
\node [square, fill=red!100]  (24,6) at (24*\m,-6*\m) {$  $}; 
   
\node [square, fill=red!100]  (25,4) at (25*\m,-4*\m) {$  $};   
\node [square, fill=red!100]  (25,5) at (25*\m,-5*\m) {$  $};    
\node [square, fill=red!100]  (25,6) at (25*\m,-6*\m) {$  $}; 

\node [square, fill=red!100]  (26,3) at (26*\m,-3*\m) {$  $};   
\node [square, fill=red!100]  (26,4) at (26*\m,-4*\m) {$  $};

\node [] (1, \h) at (-0.6*\m, -3.5*\m) {\phantom{$ \ldots $}};
\node [] (\w + 1, \h) at (\w*\m + \m + 0.6*\m, -3.5*\m) {$ \ldots $};

\end{tikzpicture}

\end{tabular}

\vspace{2em}

\begin{tabular}{c@{\extracolsep{1cm}}c}
\begin{tikzpicture}[
square/.style = {draw, rectangle, 
                 minimum size=\m, outer sep=0, inner sep=0, font=\small,
                 },
                        ]
\def\m{13pt}
\def\w{26} 
\def\wm{25}
\def\h{6}
\def\glob{3}
\def\loc{5}
    \pgfmathsetmacro\uw{int(\w/2)}
    \pgfmathsetmacro\uh{int(\h/2)}

\foreach \x in {0,...,3}
{
  \node [] at (\x*\m + \m, 0.01) {$ v_{\x} $};
}

  \foreach \x in {1,...,4}
  {
    \foreach \y in {1,...,\h}
       {    
           
           \ifnum\y>\loc
               \node [square, fill=gray!30]  (\x,\y) at (\x*\m,-\y*\m) {$  $};
           \else
               \ifnum\y>\glob
                   \node [square, fill=gray!60]  (\x,\y) at (\x*\m,-\y*\m) {$  $};
               \else
                   \node [square, fill=white]  (\x,\y) at (\x*\m,-\y*\m) {$  $};
               \fi
           \fi
       }
  }
  
\node [square, fill=red!100]  (1,2) at (1*\m,-2*\m) {$  $};  
\node [square, fill=red!100]  (3,6) at (3*\m,-6*\m) {$  $};
\node [square, fill=red!100]  (4,1) at (4*\m,-1*\m) {$  $};
  
\node [] (1, \h) at (5*\m + 1*\m, -3.5*\m) {$ \ldots $};

\node [] at (7*\m + \m, 0.01) {$ v_{t - \mu} $};
\node [] at (8.5*\m + \m, 0.01) {$ \ldots $};
\node [] at (10*\m + \m, 0.01) {$ v_{t - 1} $};

  \foreach \x in {8,...,11}
  {
    \foreach \y in {1,...,\h}
       {    
           
           \ifnum\y>\loc
               \node [square, fill=gray!30]  (\x,\y) at (\x*\m,-\y*\m) {$  $};
           \else
               \ifnum\y>\glob
                   \node [square, fill=gray!60]  (\x,\y) at (\x*\m,-\y*\m) {$  $};
               \else
                   \node [square, fill=white]  (\x,\y) at (\x*\m,-\y*\m) {$  $};
               \fi
           \fi
       }
  }

\node [square, cross out, thick]  (8,4) at (8*\m,-4*\m) {$  $};
\node [square, cross out, thick]  (9,6) at (9*\m,-6*\m) {$  $};
\node [square, cross out, thick]  (10,2) at (10*\m,-2*\m) {$  $};
\node [square, cross out, thick]  (11,6) at (11*\m,-6*\m) {$  $};

\node [] at (12*\m + \m, 0.01) {$ v_{t} $};
\node [] at (13*\m + \m, 0.01) {$ v_{t + 1} $};
\node [] at (15*\m + \m, 0.01) {$ \ldots $};
\node [] at (18*\m + \m, 0.01) {$ \ldots $};
\node [] at (20*\m + \m, 0.01) {$ v_{t + j} $};

  \foreach \x in {13,...,21}
  {
    \foreach \y in {1,...,\h}
       {    
           
           \ifnum\y>\loc
               \node [square, fill=gray!30]  (\x,\y) at (\x*\m,-\y*\m) {$  $};
           \else
               \ifnum\y>\glob
                   \node [square, fill=gray!60]  (\x,\y) at (\x*\m,-\y*\m) {$  $};
               \else
                   \node [square, fill=white]  (\x,\y) at (\x*\m,-\y*\m) {$  $};
               \fi
           \fi
       }
  }

\node [square, cross out, thick]  (13,2) at (13*\m,-2*\m) {$  $}; 
\node [square, fill=red!100]  (13,3) at (13*\m,-3*\m) {$  $};   
\node [square, fill=red!100]  (13,4) at (13*\m,-4*\m) {$  $};  
\node [square, cross out, thick]  (14,2) at (14*\m,-2*\m) {$  $};   

\node [square, fill=red!100]  (15,4) at (15*\m,-4*\m) {$  $}; 
\node [square, fill=red!100]  (15,5) at (15*\m,-5*\m) {$  $}; 
\node [square, cross out, thick]  (15,6) at (15*\m,-6*\m) {$  $}; 

\node [square, fill=red!100]  (16,3) at (16*\m,-3*\m) {$  $}; 
\node [square, fill=red!100]  (16,4) at (16*\m,-4*\m) {$  $};   
\node [square, cross out, thick]  (16,5) at (16*\m,-5*\m) {$  $};
   
\node [square, cross out, thick]  (17,6) at (17*\m,-6*\m) {$  $};
\node [square, cross out, thick]  (18,3) at (18*\m,-3*\m) {$  $};
\node [square, cross out, thick]  (19,6) at (19*\m,-6*\m) {$  $};   
\node [square, cross out, thick]  (20,3) at (20*\m,-3*\m) {$  $};

\node [square, fill=red!100]  (18,2) at (18*\m,-2*\m) {$  $};   
\node [square, fill=red!100]  (20,2) at (20*\m,-2*\m) {$  $};   
   
\node [square, fill=red!100]  (21,3) at (21*\m,-3*\m) {$  $}; 
\node [square, fill=red!100]  (21,4) at (21*\m,-4*\m) {$  $};   
\node [square, cross out, thick]  (21,5) at (21*\m,-5*\m) {$  $};

\node [] at (24.5*\m + \m, 0.01) {$ \ldots $};  
\node [] at (22*\m + \m, 0.01) {$ v_{t + j + 1} $};  

  \foreach \x in {23,...,26}
  {
    \foreach \y in {1,...,\h}
       {    
           
           \ifnum\y>\loc
               \node [square, fill=gray!30]  (\x,\y) at (\x*\m,-\y*\m) {$  $};
           \else
               \ifnum\y>\glob
                   \node [square, fill=gray!60]  (\x,\y) at (\x*\m,-\y*\m) {$  $};
               \else
                   \node [square, fill=white]  (\x,\y) at (\x*\m,-\y*\m) {$  $};
               \fi
           \fi
       }
  }   
\node [square, fill=red!100]  (23,2) at (23*\m,-2*\m) {$  $};   
\node [square, fill=red!100]  (23,3) at (23*\m,-3*\m) {$  $};      
   
\node [square, fill=red!100]  (25,4) at (25*\m,-4*\m) {$  $};   
\node [square, fill=red!100]  (25,5) at (25*\m,-5*\m) {$  $};   

\node [square, fill=red!100]  (26,3) at (26*\m,-3*\m) {$  $};   
\node [square, fill=red!100]  (26,4) at (26*\m,-4*\m) {$  $};

\node [square, fill=white!100]  (23,1) at (23*\m,-1*\m) {$  $}; 
\node [square, fill=gray!60]  (26,4) at (26*\m,-4*\m) {$  $};  
\node [square, cross out, thick, thick]  (23,1) at (23*\m,-1*\m) {$  $}; 
\node [square, cross out, thick]  (24,6) at (24*\m,-6*\m) {$  $}; 
\node [square, cross out, thick]  (25,6) at (25*\m,-6*\m) {$  $};  
\node [square, cross out, thick]  (26,4) at (26*\m,-4*\m) {$  $};

\node [] (1, \h) at (-0.6*\m, -3.5*\m) {\phantom{$ \ldots $}};

\node [] (\w + 1, \h) at (\w*\m + \m + 0.6*\m, -3.5*\m) {$ \ldots $};

   \draw [draw, decorate, thick,decoration={brace,amplitude=5pt,mirror}]
   (0.5*\m,-7*\m) -- (4.5*\m,-7*\m) node[midway,yshift=-2em]{\begin{tabular}{c}Repaired locally\\ (column-wise) \end{tabular}};        
   \draw [draw, decorate, thick,decoration={brace,amplitude=5pt,mirror}]
   (7.5*\m,-7*\m) -- (11.5*\m,-7*\m) node[midway,yshift=-2em]{\begin{tabular}{c}Needed for\\ $ v_t, v_{t+1}, \ldots $ \end{tabular}};       
   \draw [draw, decorate, thick,decoration={brace,amplitude=5pt,mirror}]
   (12.5*\m,-7*\m) -- (21.5*\m,-7*\m) node[midway,yshift=-2em]{\begin{tabular}{c} $ < (N-k)(j+1) $ \\ erasures after restriction \end{tabular}};        
   \draw [draw, decorate, thick,decoration={brace,amplitude=5pt,mirror}]
   (22.5*\m,-7*\m) -- (26.5*\m,-7*\m) node[midway,yshift=-2em]{\begin{tabular}{c}To be repaired\\ in next windows \end{tabular}};

\end{tikzpicture}

\end{tabular}
\end{center}

\caption{Sliding-window repair in a partial $ j $-MDS code, with parameters as in Fig. \ref{fig sliding window repair} (thus $ N = 5 $). Following Definition \ref{def partial MDS} and assuming an erasure pattern as in Fig. \ref{fig sliding window repair} (upper figure), we may remove one local parity in each block (depicted as a crossed box) before proceeding with the sliding-window repair for the restricted $ j $-MDS code, which would be performed as in Theorem \ref{th sliding window} and may correct up to $ (N-k)(j+1) $ erasures after removing \textit{arbitrary} local parities. Thus a local parity to be removed should be chosen as one of the erased symbols in case there are erased symbols in the corresponding block. Otherwise, we may remove any local parity. After the catastrophic erasures are corrected, the local parities are added again locally (column-wise) in case they were erased. }
\label{fig sliding window repair in partial MDP}
\end{figure*}

Some explanations about Definition \ref{def partial MDS} are in order.

First, we observe that the restricted convolutional code $ \mathcal{C}_{\Delta} \subseteq \mathbb{F}[D]^N $ in the previous definition has rank $ k $ by Lemma \ref{lemma convolutional if local redund removed} below, thus the definition is consistent. 

\begin{lemma} \label{lemma convolutional if local redund removed}
Let $ \mathcal{C} \subseteq \mathbb{F}[D]^n $ be an $ (n,k,r,\partial) $ LRCC with local groups $ \Gamma_i $, for $ i = 1,2, \ldots, g $, as in Definition \ref{def lrc convolutional codes}. Let $ \Delta_i \subseteq \Gamma_i $ be such that $ | \Gamma_i \setminus \Delta_i | \leq \partial - 1 $, for $ i = 1,2, \ldots, g $, and define $ \Delta = \bigcup_{i=1}^g \Delta_i $ and $ N = | \Delta | $. Then the restricted code $ \mathcal{C}_\Delta \subseteq \mathbb{F}[D]^N $ has rank $ k $, or in other words, it is an $ (N,k) $ convolutional code.

In addition, if $ G(D) = \sum_{j=0}^\mu G_jD^j $ is a reduced generator matrix of $ \mathcal{C} $ such that $ G_0 $ is full-rank, then $ (G_0)_\Delta $ is also full-rank. 
\end{lemma}
\begin{proof}
Let $ G(D) \in \mathbb{F}[D]^{k \times n} $ be a generator matrix of $ \mathcal{C} $. It suffices to prove that the rows of $ G(D)_\Delta \in \mathbb{F}[D]^{k \times N} $ are $ \mathbb{F}[D] $-linearly independent. 

Assume that there exists $ x(D) \in \mathbb{F}[D]^k $ such that $ x(D) G(D)_\Delta = 0 $. If $ v(D) = x(D)G(D) \in \mathcal{C} $, then we have that $ v(D)_\Delta = x(D) G(D)_\Delta = 0 $. Write $ v(D) = \sum_{j \in \mathbb{N}} v_j D^j $ and fix $ j \in \mathbb{N} $. We then deduce that $ (v_j)_\Delta = 0 $, and therefore $ (v_j)_{\Delta_i} = 0 $, for $ i = 1,2, \ldots, g $. Since $ \dd(\mathcal{C}_{\Gamma_i}^0) \geq \partial $ and $ | \Gamma_i \setminus \Delta_i | \leq \partial - 1 $, we deduce that $ (v_j)_{\Gamma_i} = 0 $, for $ i = 1,2, \ldots, g $. Now, because $ [n] = \bigcup_{i=1}^g \Gamma_i $, we conclude that $ v_j = 0 $. 

Thus we have proven that $ x(D)G(D) = 0 $. Since $ G(D) $ has full rank, we conclude that $ x(D) = 0 $, and we are done. The statement regarding $ G_0 $ and $ (G_0)_\Delta $ is proven following the same lines.  
\end{proof}

Similar to the case of block codes (replacing $ j $-MDS by MDS), the term partial $ j $-MDS is motivated by the fact that the column distances attain the bound (\ref{eq singleton bound with localities}) (see Proposition \ref{prop partial MDS implies optimal distance} below), thus they have smaller column distances than those of $ j $-MDS codes (this is the price to pay for locality). However, partial $ j $-MDS codes as in Definition \ref{def partial MDS} can be seen as $ j $-MDS codes that can be \textit{added} locality in some optimal sense: We can recover some other $ j $-MDS code after removing \textit{any} (maximal) collection of local parities, not only the added ones. Due to this reason, we gain a considerable flexibility in the erasure patterns that can be corrected (see Fig. \ref{fig sliding window repair in partial MDP}). 

In the block case, partial MDS codes can be equivalently defined as follows: A locally repairable block code is partial MDS if it can correct all erasure patterns that are information-theoretically correctable for the given local constraints $ r $ and $ \partial $ and the given dimension $ k $ and length $ n $. Obviously, if there are no local constraints ($ \partial = 1 $ for instance), then being able to correct all information-theoretically correctable erasure patterns is equivalent to being MDS.  

See Fig. \ref{fig sliding window repair in partial MDP} for a graphical description of sliding-window repair combined with local repair in a partial $ j $-MDS convolutional code.

We now show that partial $ j $-MDS codes attain the bound (\ref{eq singleton bound with localities}), hence being optimal LRCCs in terms of column distances. We need a preliminary lemma, which is of interest by itself and which follows directly from Definition \ref{def partial MDS} and Proposition \ref{prop monotonicity singleton column distance}.

\begin{lemma}
If an LRCC is partial $ j $-MDS, then it is partial $ h $-MDS, for all $ h = 0,1,2, \ldots, j $.
\end{lemma}

\begin{proposition} \label{prop partial MDS implies optimal distance}
If an $ (n,k,r,\partial) $ LRCC is partial $ j $-MDS for some $ j \in \mathbb{N} $, then its column distances attain the bound (\ref{eq singleton bound with localities}), for all $ h = 0,1,2, \ldots, j $.
\end{proposition}
\begin{proof}
By the previous lemma, we only need to prove the result for $ h = j $. For such a case, the proof follows exactly the same lines as the proof of Theorem \ref{th MSRD implies optimal LRCC}, and is left to the reader.
\end{proof}

\begin{remark} \label{remark optimal LRCs are not MR LRCs}
In the block case, the converse is not true. For instance, Tamo-Barg codes \cite{tamo-barg} are locally repairable codes with optimal global distance, but cannot always be maximally recoverable (partial MDS) by the field-size bound in \cite[Eq. (2)]{gopi}. We conjecture, but do not prove or disprove, that not every LRCC attaining the bound (\ref{eq singleton bound with localities}), for some $ j \in \mathbb{N} $, is a partial $ j $-MDS convolutional code.
\end{remark}

Our next goal is to define partial MDP convolutional codes, which are partial $ j $-MDS for the maximum value of $ j $. We first need the following lemma, which follows directly from the definitions and Proposition \ref{prop monotonicity singleton column distance}.

\begin{lemma} \label{lemma degree of restricted conv code}
Let $ \mathcal{C} \subseteq \mathbb{F}[D]^n $ be an $ (n,k) $ convolutional code. For any $ \Delta \subseteq [n] $, it holds that
$$ \delta(\mathcal{C}_\Delta) \leq \delta(\mathcal{C}). $$
In particular, if $ \mathcal{C}_\Delta $ is $ j $-MDS, then $ j \leq \left\lfloor \frac{\delta}{k} \right\rfloor + \left\lfloor \frac{\delta}{N-k} \right\rfloor $, for $ N = |\Delta| $ and $ \delta = \delta(\mathcal{C}) $.
\end{lemma}

We may now define partial MDP convolutional codes.

\begin{definition} \label{def partial MDP}
We say that an $ (n,k,r,\partial) $ LRCC $ \mathcal{C} \subseteq \mathbb{F}[D]^n $ is \textit{partial MDP} if it is partial $ L $-MDS for $ L = \lfloor \frac{\delta}{k} \rfloor + \lfloor \frac{\delta}{N-k} \rfloor $, where $ N = n - g(\partial-1) $ and $ \delta = \delta(\mathcal{C}) $.
\end{definition}

The main purpose of this section is to show that the global code in Construction \ref{construction} based on an MSRD outer code (for instance, that in Appendix \ref{app MSRD conv codes}) is partial MDP. In particular, we will show the existence of partial MDP codes for general parameters, over any characteristic, for sufficiently large fields. 

We first need the following lemma, which is \cite[Th. 1]{lrc-sum-rank}. Observe that we will make use of this lemma in the non-linear case.

\begin{lemma} [\textbf{\cite{lrc-sum-rank}}] \label{lemma SR distance as minimum among H distances}
Recall that $ N = gr $. Given a (linear or non-linear) block code $ \mathcal{C} \subseteq \mathbb{F}_{q^m}^N $, it holds that
$$ \dd_{SR}(\mathcal{C}) = \min \left\lbrace \dd(\mathcal{C} \diag(  B_1, B_2, \ldots, B_g  )) \mid  B_i  \in \mathbb{F}_q^{r \times r} \textrm{ invertible}, i = 1,2, \ldots, g \right\rbrace. $$
\end{lemma}

We may now prove the main result of this section.

\begin{theorem} \label{th MSRD implies partial MDS}
In Construction \ref{construction}, the following hold:
\begin{enumerate}
\item
If $ j \in \mathbb{N} $ and $ \mathcal{C}_{out} $ is non-catastrophic and $ j $-MSRD, then $ \mathcal{C}_{glob} $ is partial $ j $-MDS.
\item
$ \delta(\mathcal{C}_{glob}) = \delta(\mathcal{C}_{out}) $ and $ \mu(\mathcal{C}_{glob}) = \mu(\mathcal{C}_{out}) $.
\item
If $ \mathcal{C}_{out} $ is non-catastrophic and $ L $-MSRD, where $ L = \lfloor \frac{\delta}{k} \rfloor + \lfloor \frac{\delta}{N-k} \rfloor $ and $ \delta = \delta(\mathcal{C}_{out}) $, then $ \mathcal{C}_{glob} $ is partial MDP.
\end{enumerate}
\end{theorem}
\begin{proof}
We start by proving Item 1. Let $ \Delta_i \subseteq \Gamma_i $ be such that $ |\Gamma_i \setminus \Delta_i | = \partial - 1 $ (i.e. $ |\Delta_i| = r $ since $ |\Gamma_i| = r+\partial-1 $ in Construction \ref{construction}), for $ i = 1,2, \ldots, g $. If $ \Delta = \bigcup_{i=1}^g \Delta_i $ and $ N = |\Delta| $, then the restricted code $ (\mathcal{C}_{glob})_\Delta \subseteq \mathbb{F}_{q^m}[D]^N $ is the $ (N,k) $ convolutional code given by
\begin{equation}
(\mathcal{C}_{glob})_\Delta = \mathcal{C}_{out} \diag(A_{\Delta_1}, A_{\Delta_2}, \ldots, A_{\Delta_g}) \subseteq \mathbb{F}_{q^m}[D]^N.
\label{eq in proof of partial MDS}
\end{equation}
Since $ \mathcal{C}_{loc} \subseteq \mathbb{F}_q^{r+\partial-1} $ is an $ (r+\partial-1, r) $ MDS block code and $ | \Delta_i | = r $, we deduce that $ A|_{\Delta_i} \in \mathbb{F}_q^{r \times r} $ is invertible, for $ i = 1,2, \ldots, g $. Thus $ (\mathcal{C}_{glob})_\Delta $ is non-catastrophic by Lemma \ref{lemma restricted const 1 is non-catastrophic}, and moreover by (\ref{eq in proof of partial MDS}) and Lemma \ref{lemma SR distance as minimum among H distances}, we have that
$$ \dd_j^c((\mathcal{C}_{glob})_\Delta) \geq \dd_{SR,j}^c((\mathcal{C}_{glob})_\Delta) = \dd_{SR,j}^c(\mathcal{C}_{out}) = (N-k)(j+1) + 1. $$
Hence $ (\mathcal{C}_{glob})_\Delta $ is $ j $-MDS and Item 1 follows.

Now, Item 2 follows from the fact that $ \mathcal{C}_{glob} = \mathcal{C}_{out} \diag_g(A) $, and multiplying by the full-rank constant matrix $ \diag_g(A) \in \mathbb{F}_{q^m}^{N \times n} \subseteq \mathbb{F}_{q^m}[D]^{N \times n} $ on the right preserves degrees. Finally, Item 3 follows by combining Items 1 and 2.
\end{proof}

Finally, by plugging in the previous theorem the $ L $-MSRD codes from \cite{mahmood-convolutional} (see Appendix \ref{app MSRD conv codes}) as outer codes in Construction \ref{construction}, we show the existence of partial MDP convolutional codes.

\begin{corollary} \label{cor construction of partial MDP}
If $ N = gr $, $ (N-k) | \delta $, $ M = \max \{ N-k,k \} $, $ L = \lfloor \frac{\delta}{k} \rfloor + \delta/(N-k) $, $ q \geq r+\partial - 1 $ and $ m \geq q^{M(L+2) - 1} $, then the convolutional code from Corollary \ref{cor j-MSRD applied to const 1} is an $ (n,k,r,\partial) $ partial MDP convolutional code.
\end{corollary}

Observe that we could have given the previous corollary first, and then deduce Corollary \ref{cor j-MSRD applied to const 1} from Proposition \ref{prop partial MDS implies optimal distance}. However, we have chosen to present our results in this order for simplicity.

\section{Further Considerations} \label{sec further considerations}

\subsection{Unequal Localities and Local Distances} \label{subsec unequal}

Locally repairable codes with unequal localities were introduced independently in \cite{kadhe, zeh-multiple}. Adding also unequal local distances was first considered in \cite{chen-hao}. Essentially, locally repairable codes with unequal localities are those such that the locality $ r $ and local distance $ \partial $ depend on the local group $ \Gamma_i $ (see Definition \ref{def lrc convolutional codes}). In other words, the $ i $th local group has locality $ r_i $ and local distance $ \partial_i $, for $ i = 1,2, \ldots, g $. We may then modify Definition \ref{def lrc convolutional codes} to include unequal localities and local distances by adding indices to Items 1 and 2:
\begin{enumerate}
\item
$ |\Gamma_i| \leq r_i + \partial_i - 1 $,
\item
$ \dd(\mathcal{C}_{\Gamma_i}^0) \geq \partial_i $,
\end{enumerate}
for $ i = 1,2, \ldots, g $. The main motivation for this type of locally repairable codes is that some nodes may require faster repair or access (\textit{hot data}), while considering the different localities in general improves the global correction capability of the code.

Finding analogous upper bounds to (\ref{eq singleton bound with localities}) is a challenging task in general. Such bounds are known when $ r_1 \leq r_2 \leq \ldots \leq r_g $ and $ \partial_1 \geq \partial_2 \geq \ldots \geq \partial_g $ (see \cite[Th. 2]{chen-hao} and \cite[Th. 2]{kim}). 

On the other hand, adapting the notion of partial MDS codes to unequal localities is straightforward (see \cite[Def. 5]{lrc-sum-rank}). In addition, it was proven in \cite[Th. 2]{lrc-sum-rank} that MSRD block codes used as outer codes always give partial MDS codes, for any choice of unequal localities and local distances. 

All the results in this work hold also for unequal localities and local distances. As in the block case, bounds on the column distances are not straightforward in general. However, Construction \ref{construction} with the MSRD codes from Appendix \ref{app MSRD conv codes} as outer codes provide partial MDP codes for an arbitrary choice of unequal localities and local distances, just as in the block case. We leave the details to the reader.

\subsection{Tail-Biting LRCCs} \label{subsec tail-biting LRCC}

LRCCs may encode an unrestricted number of information symbols (i.e. files or file components), while locality and sliding-window erasure-correction capability and complexity remain constant. However, truncating an $ (n,k) $ LRCC $ \mathcal{C} $ at a given block $ t $ implies that, for $ h \in \mathbb{N} $, the final windows $ (v_{t-h}, v_{t-h+1}, \ldots, v_t) $ cannot be the initial part of a sliding window consisting of $ j+1 > h+1 $ blocks, which could potentially correct $ \dd^c_j(\mathcal{C}) -1 > \dd^c_h(\mathcal{C}) -1 $ erasures. Therefore, in such a truncated LRCC, certain blocks receive a weaker protection against erasures. 

To provide equal protection to all blocks, one solution is to terminate the LRCC as a block code by converting it into a \textit{tail-biting} convolutional code. This simply requires updating the first $ \mu $ blocks using the last $ \mu $, in the way they would be encoded if we had used the generator matrix
$$ \left[ \begin{array}{ccccccccc}
G_0 & G_1 & \ldots & G_\mu & \ldots &  & & & \\
 & G_0 & \ldots & G_{\mu-1} & \ldots &  &  & & \\
 & & \ddots & \vdots & \ldots &  &  & & \\
 & & & G_0 & \ldots &  &  &  & \\
 & & & & \ddots & \ddots & \ddots & \ddots & \\
 & & & & \ldots & G_0 & G_1 & \ldots & G_\mu \\
G_\mu & & & & \ldots & & G_0 & \ldots & G_{\mu-1} \\
\vdots & \vdots & \ddots & & \ldots & & & \ddots & \vdots \\
G_1 & G_2 & \ldots & & \ldots & & & & G_0 
\end{array} \right], $$
where $ G(D) = \sum_{j=0}^\mu G_j D^j \in \mathbb{F}[D]^{k \times n} $ is a reduced generator matrix of the LRCC. In this way, sliding-window repair behaves equally in any window of the same size. However, we always need to have at least $ \mu $ consecutive blocks with no erasures in order to get the repair started, although this $ \mu $ consecutive blocks may be arbitrary and not necessarily the first $ \mu $. In other words, any $ \mu $ consecutive blocks may be considered initial in a tail-biting convolutional code.

\section*{Acknowledgement}

The first author is supported by The Independent Research Fund Denmark (Grant No. DFF-7027-00053B). The second author is partially supported by the Generalitat Valenciana (Grant No. AICO/2017/128) and the Universitat d'Alacant (Grant No. VIGROB-287).


\begin{thebibliography}{10}

\bibitem{almeida}
P.~Almeida, D.~Napp, and R.~Pinto.
\newblock A new class of superregular matrices and {MDP} convolutional codes.
\newblock {\em Linear Algebra and its Applications}, 439(7):2145--2157, 2013.

\bibitem{al16}
P.~Almeida, D.~Napp, and R.~Pinto.
\newblock Superregular matrices and applications to convolutional codes.
\newblock {\em Linear Algebra and its Applications}, 499:1--25, 2016.

\bibitem{fountainLRC}
M.~Asteris and A.~G. Dimakis.
\newblock Repairable fountain codes.
\newblock {\em IEEE J.\ Select.\ Areas Comm.}, 32(5):1037--1047, May 2014.

\bibitem{blaum-RAID}
M.~Blaum, J.~L. Hafner, and S.~Hetzler.
\newblock Partial-{MDS} codes and their application to {RAID} type of
  architectures.
\newblock {\em IEEE Trans.\ Info.\ Theory}, 59(7):4510--4519, July 2013.

\bibitem{fountain}
J.~W. Byers, M.~Luby, M.~Mitzenmacher, and A.~Rege.
\newblock A {D}igital {F}ountain approach to reliable distribution of bulk
  data.
\newblock {\em SIGCOMM Comput. Commun. Rev.}, 28(4):56--67, October 1998.

\bibitem{chen-hao}
B.~Chen, S.~T. Xia, and J.~Hao.
\newblock Locally repairable codes with multiple $ (r_i,\delta_i) $-localities.
\newblock In {\em Proc.\ IEEE Int.\ Symp. Info.\ Theory}, pages 2038--2042,
  June 2017.

\bibitem{lrc-convoluted}
A.~Datta.
\newblock Locally repairable rapid{RAID} systematic codes — one simple
  convoluted way to get it all.
\newblock In {\em Proc.\ IEEE Info. Theory Workshop}, pages 60--64, Nov 2014.

\bibitem{gabrys}
R.~{Gabrys}, E.~{Yaakobi}, M.~{Blaum}, and P.~H. {Siegel}.
\newblock Constructions of partial {MDS} codes over small fields.
\newblock {\em IEEE Trans.\ Info.\ Theory}, 65(6):3692--3701, Dec 2018.

\bibitem{stronglyMDS}
H.~Gluesing-Luerssen, J.~Rosenthal, and R.~Smarandache.
\newblock Strongly-{MDS} convolutional codes.
\newblock {\em IEEE Trans.\ Info.\ Theory}, 52(2):584--598, Feb 2006.

\bibitem{gopalan-MR}
P.~Gopalan, C.~Huang, B.~Jenkins, and S.~Yekhanin.
\newblock Explicit maximally recoverable codes with locality.
\newblock {\em IEEE Trans.\ Info.\ Theory}, 60(9):5245--5256, Sept 2014.

\bibitem{gopalan}
P.~Gopalan, C.~Huang, H.~Simitci, and S.~Yekhanin.
\newblock On the locality of codeword symbols.
\newblock {\em IEEE Trans.\ Info.\ Theory}, 58(11):6925--6934, Nov 2012.

\bibitem{gopi}
S.~Gopi, V.~Guruswami, and S.~Yekhanin.
\newblock On maximally recoverable local reconstruction codes.
\newblock {\em Electr. Colloq. Comp. Complexity (ECCC)}, 24(183), 2017.

\bibitem{azure}
C.~Huang, H.~Simitci, Y.~Xu, A.~Ogus, B.~Calder, P.~Gopalan, J.~Li, and
  S.~Yekhanin.
\newblock Erasure coding in {W}indows {A}zure storage.
\newblock In {\em 2012 {USENIX} Annual Technical Conference}, pages 15--26,
  Boston, MA, 2012.

\bibitem{ivanov}
F.~{Ivanov}, A.~{Kreshchuk}, and V.~{Zyablov}.
\newblock On the local erasure correction capacity of convolutional codes.
\newblock In {\em 2018 Int. Symp. Info. Theory and Applications (ISITA)}, pages
  296--300, Oct 2018.

\bibitem{kadhe}
S.~Kadhe and A.~Sprintson.
\newblock Codes with unequal locality.
\newblock In {\em Proc.\ IEEE Int.\ Symp. Info.\ Theory}, pages 435--439, July
  2016.

\bibitem{kamath}
G.~M. Kamath, N.~Prakash, V.~Lalitha, and P.~V. Kumar.
\newblock Codes with local regeneration and erasure correction.
\newblock {\em IEEE Trans.\ Info.\ Theory}, 60(8):4637--4660, Aug 2014.

\bibitem{kim}
G.~Kim and J.~Lee.
\newblock Locally repairable codes with unequal local erasure correction.
\newblock {\em IEEE Trans.\ Info.\ Theory}, 64(11):7137--7152, May 2018.

\bibitem{primitive-normal}
H.~W. Lenstra and R.~J. Schoof.
\newblock Primitive normal bases for finite fields.
\newblock {\em Mathematics of Computation}, 48(177):217--231, 1987.

\bibitem{space-time-kumar}
H.-F. Lu and P.~V. Kumar.
\newblock A unified construction of space-time codes with optimal
  rate-diversity tradeoff.
\newblock {\em IEEE Trans.\ Info.\ Theory}, 51(5):1709--1730, May 2005.

\bibitem{mackay}
D.~J.~C. MacKay.
\newblock {\em Information Theory, Inference \& Learning Algorithms}.
\newblock Cambridge University Press, New York, NY, USA, 2002.

\bibitem{mahmood-convolutional}
R.~Mahmood, A.~Badr, and A.~Khisti.
\newblock Convolutional codes with maximum column sum rank for network
  streaming.
\newblock {\em IEEE Trans.\ Info.\ Theory}, 62(6):3039--3052, 2016.

\bibitem{linearizedRS}
U.~Mart{\'i}nez-Pe{\~n}as.
\newblock Skew and linearized {R}eed-{S}olomon codes and maximum sum rank
  distance codes over any division ring.
\newblock {\em J.\ Algebra}, 504:587--612, 2018.

\bibitem{secure-multishot}
U.~Mart{\'i}nez-Pe{\~n}as and F.~R. Kschischang.
\newblock Reliable and secure multishot network coding using linearized
  {R}eed-{S}olomon codes.
\newblock {\em IEEE Trans.\ Info.\ Theory}, 65(8):4785--4803, Aug 2019.

\bibitem{lrc-sum-rank}
U.~Mart{\'i}nez-Pe{\~n}as and F.~R. Kschischang.
\newblock Universal and dynamic locally repairable codes with maximal
  recoverability via sum-rank codes.
\newblock {\em IEEE Trans.\ Info.\ Theory}, pages 1--1, 2019.

\bibitem{general-mceliece}
R.~J. McEliece and R.~P. Stanley.
\newblock The general theory of convolutional codes.
\newblock {\em The Telecommunications and Data Acquisition Report},
  42(113):89--98, 1993.

\bibitem{mrd-convolutional}
D.~Napp, R.~Pinto, J.~Rosenthal, and P.~Vettori.
\newblock {MRD} rank metric convolutional codes.
\newblock In {\em Proc.\ IEEE Int.\ Symp. Info.\ Theory}, pages 2766--2770,
  2017.

\bibitem{NaRo2015}
D.~Napp and R.~Smarandache.
\newblock Constructing strongly {MDS} convolutional codes with maximum distance
  profile.
\newblock {\em Advances in Mathematics of Communications}, 10(2):275--290,
  2016.

\bibitem{multishot}
R.~W. N{\'o}brega and B.~F. Uch{\^o}a-Filho.
\newblock Multishot codes for network coding using rank-metric codes.
\newblock In {\em Proc. 2010 Third IEEE Int.\ Workshop on Wireless Network
  Coding}, pages 1--6, 2010.

\bibitem{rawat}
A.~S. Rawat, O.~O. Koyluoglu, N.~Silberstein, and S.~Vishwanath.
\newblock Optimal locally repairable and secure codes for distributed storage
  systems.
\newblock {\em IEEE Trans.\ Info.\ Theory}, 60(1):212--236, 2014.

\bibitem{MDSconvolutional}
J.~Rosenthal and R.~Smarandache.
\newblock Maximum {D}istance {S}eparable {C}onvolutional {C}odes.
\newblock {\em Applicable Algebra in Engineering, Communication and Computing},
  10(1):15--32, Aug 1999.

\bibitem{xoring}
M.~Sathiamoorthy, M.~Asteris, D.~Papailiopoulos, A.~G. Dimakis, R.~Vadali,
  S.~Chen, and D.~Borthakur.
\newblock {XOR}ing elephants: novel erasure codes for big data.
\newblock In {\em Proc. 39th int. conf. {V}ery {L}arge {D}ata {B}ases},
  PVLDB'13, pages 325--336, 2013.

\bibitem{tamo-barg}
I.~Tamo and A.~Barg.
\newblock A family of optimal locally recoverable codes.
\newblock {\em IEEE Trans.\ Info.\ Theory}, 60(8):4661--4676, Aug 2014.

\bibitem{erasure-convolutional}
V.~Tomas, J.~Rosenthal, and R.~Smarandache.
\newblock Decoding of convolutional codes over the erasure channel.
\newblock {\em IEEE Trans.\ Info.\ Theory}, 58(1):90--108, Jan 2012.

\bibitem{wachter-convolutional}
A.~Wachter-Zeh, M.~Stinner, and V.~Sidorenko.
\newblock Convolutional codes in rank metric with application to random network
  coding.
\newblock {\em IEEE Trans.\ Info.\ Theory}, 61(6):3199--3213, 2015.

\bibitem{zeh-multiple}
A.~Zeh and E.~Yaakobi.
\newblock Bounds and constructions of codes with multiple localities.
\newblock In {\em Proc.\ IEEE Int.\ Symp. Info.\ Theory}, pages 640--644, July
  2016.

\bibitem{rcc}
B.~Zhu, X.~Li, H.~Li, and K.~W. Shum.
\newblock Replicated convolutional codes: A design framework for
  repair-efficient distributed storage codes.
\newblock In {\em Proc.\ Allerton Conf. Comm. Control Comp.}, pages 1018--1024,
  Sept 2016.

\end{thebibliography}

\appendix

\section{A Lemma on Information Sets of Optimal Block LRCs} \label{app lemma optimal block lrc}

In this appendix, we prove the following result on the information sets of optimal block LRCs. Essentially, we follow a simplified version of the proof of \cite[Th. 21]{rawat}, using linear LRCs, thus dimensions instead of entropies, and for pair-wise disjoint local groups of size exactly $ r + \partial - 1 $, for $ (r, \partial) $-localities.

\begin{lemma} \label{lemma info sets optimal block lrc}
Let $ \mathcal{C}_0 \subseteq \mathbb{F}^n $ be a $ k $-dimensional block linear LRC with $ (r, \partial) $-localities, where we are considering pair-wise disjoint local groups of size exactly $ r+\partial - 1 $: $ [n] = \Gamma_1 \cup \Gamma_2 \cup \ldots \cup \Gamma_g $, where $ \Gamma_i \cap \Gamma_j = \varnothing $ if $ i \neq j $, and $ | \Gamma_i | = r + \partial - 1 $, for $ i = 1,2, \ldots, g $. Define $ \ell = \left\lceil k / r \right\rceil $, and assume that $ \mathcal{C}_0 $ has maximum possible minimum Hamming distance, i.e., 
$$ \dd(\mathcal{C}_0) = (n-k) - \left( \ell - 1 \right)(\partial - 1) + 1. $$
Then there are $ \ell $ local groups, which we may assume without loss of generality that they are the first $ \ell $ of them, $ \Gamma_1, \Gamma_2, \ldots, \Gamma_\ell $, such that 
$$ \dim \left( (\mathcal{C}_0)_{\Gamma_1 \cup \Gamma_2 \cup \ldots \cup \Gamma_\ell} \right) = k, $$
where $ \mathcal{C}_\Gamma \subseteq \mathbb{F}^{|\Gamma|} $ denotes the restriction of a block code $ \mathcal{C} \subseteq \mathbb{F}^n $ onto the coordinates in $ \Gamma \subseteq [n] $.
\end{lemma}
\begin{proof}
We proceed as in the proof of \cite[Th. 21]{rawat}, and define the following algorithm, which finds a size-$ \ell $ subset $ \mathcal{I} \subseteq [g] $ of local groups satisfying the properties in the lemma. As explained above, this algorithm is the same as that in the proof of \cite[Th. 21]{rawat}, but considering only linear LRCs, replacing entropies by dimensions, and considering pair-wise disjoint local groups of size exactly $ r + \partial - 1 $, for $ (r, \partial) $-localities.

\begin{algorithmic} [1]
\STATE Set $ \mathcal{I} = \varnothing $ and $ \mathcal{A} = \varnothing $.
\WHILE {$ \dim \left( (\mathcal{C}_0)_{\mathcal{A}} \right) < k $} 
{\STATE{Pick an index $ i \in [g] \setminus \mathcal{I} $.}
	\IF{ $ \dim \left( (\mathcal{C}_0)_{\mathcal{A} \cup \Gamma_i} \right) < k $ }
	\STATE{ Set $ \mathcal{I} := \mathcal{I} \cup \{ i \} $. }
	\STATE{ Set $ \mathcal{A} := \mathcal{A} \cup \Gamma_i $. }
	\ELSIF{ $ \dim \left( (\mathcal{C}_0)_{\mathcal{A} \cup \Gamma_i} \right) \geq k $ and $ \exists \Delta \subseteq \Gamma_i $ s.t. $ \dim \left( (\mathcal{C}_0)_{\mathcal{A} \cup \Delta } \right) < k $ }
	\STATE{ Set $ \mathcal{I} := \mathcal{I} \cup \{ i \} $. }
	\STATE{ Set $ \mathcal{A} := \mathcal{A} \cup \Delta $. }
	\ELSE
	\STATE{ \textbf{end while} }
	\ENDIF }
\ENDWHILE
\RETURN $ \mathcal{I}, \mathcal{A} $
\end{algorithmic}

Now we run the algorithm above. As in the proof of \cite[Th. 21]{rawat}, there may be only the following two cases.  

\textbf{Case 1:} Assume that the algorithm terminates with the final sets $ \mathcal{I} $ and $ \mathcal{A} $ assigned at lines 5 and 6, respectively. Since the algorithm has terminated at this point, if we consider any $ i \in [g] \setminus \mathcal{I} $, then
$$ \dim \left( (\mathcal{C}_0)_{\mathcal{A} \cup \Gamma_i} \right) \geq k . $$
Hence we reassign $ \mathcal{I} := \mathcal{I} \cup \{ i \} $, and then it must hold that
$$ | \mathcal{I} | \geq \left\lceil \frac{k}{r} \right\rceil = \ell , $$
since adding a local group may not increase the dimension of the restricted code by more than $ r $, since $ \partial - 1 $ out of $ r + \partial - 1 $ coordinates in a local group are redundant.

\textbf{Case 2:} Assume that the algorithm terminates with the final sets $ \mathcal{I} $ and $ \mathcal{A} $ assigned at lines 8 and 9, respectively. In this case, we already have that
$$ \dim \left( (\mathcal{C}_0)_{\bigcup_{j \in \mathcal{I}} \Gamma_j} \right) \geq k , $$
thus, without reassigning $ \mathcal{I} $, we also deduce that
$$ | \mathcal{I} | \geq \left\lceil \frac{k}{r} \right\rceil = \ell . $$

In any of the two cases, Case 1 or Case 2, assume that $ | \mathcal{I} | > \ell $. Following the same steps as in the proof of \cite[Th. 21]{rawat}, we have that ($ \partial > 1 $)
\begin{equation*}
\begin{split}
\dd (\mathcal{C}_0) & \leq n - k + 1 - \left( |\mathcal{I}| - 1 \right) (\partial - 1) \\
& < n - k + 1 - (\ell - 1) (\partial - 1) \\
& = (n-k) - \left( \left\lceil \frac{k}{r} \right\rceil - 1 \right)(\partial - 1) + 1.
\end{split}
\end{equation*}
This contradicts the optimality of the LRC $ \mathcal{C}_0 $, hence the case $ | \mathcal{I} | > \ell $ may not happen. Therefore, it must hold that $ | \mathcal{I} | = \ell $, in both Case 1 and Case 2. Also in both cases, the local groups $ \Gamma_j $, for $ j \in \mathcal{I} $, satisfy the properties of the lemma, i.e., 
$$ \dim \left( (\mathcal{C}_0)_{\bigcup _{j \in \mathcal{I}} \Gamma_j} \right) = k , $$
and thus we are done.
\end{proof}

\section{Known Construction of MSRD Convolutional Codes} \label{app MSRD conv codes}

In this appendix, we revisit the construction of non-catastrophic MSRD convolutional codes from \cite{mahmood-convolutional}, which is based on the superregular matrices introduced in \cite{almeida}. To the best of our knowledge, this is the only known construction of MSRD convolutional codes. In addition, they admit general parameters, except that they usually require impractically large field sizes. Acceptable field sizes can be achieved for certain parameters. See Table I in \cite{mahmood-convolutional} for a few instances.

Fix $ 1 \leq k \leq N $. As in \cite{almeida}, see also \cite{al16}, we will restrict ourselves to $ (N,k) $ convolutional codes whose degree $ \delta $ satisfies that $ (N-k) | \delta $, for general parameters see \cite{NaRo2015}. Define $ M = \max \{ N-k, k \} $ and $ L = \lfloor \frac{\delta}{k} \rfloor + \delta / (N-k) $, as in (\ref{eq def of L}). Let $ q $ be any prime power and assume that
\begin{equation}
m \geq q^{M(L+2) - 1}.
\label{eq bound on field size}
\end{equation}
The field will be then $ \mathbb{F} = \mathbb{F}_{q^m} $. Let $ \alpha \in \mathbb{F}_{q^m} $ be a \textit{primitive normal element} over $ \mathbb{F}_q $, that is, a primitive element of $ \mathbb{F}_{q^m} $ such that $ \alpha, \alpha^q, \ldots, \alpha^{q^{m-1}} $ form a basis of $ \mathbb{F}_{q^m} $ over $ \mathbb{F}_q $. Such element exists for any finite field extension $ \mathbb{F}_q \subseteq \mathbb{F}_{q^m} $ (see \cite{primitive-normal}). Define the matrix
\begin{equation}
T_j = \left[ \begin{array}{cccc}
\alpha^{[Mj]} & \alpha^{[Mj + 1]} & \ldots & \alpha^{[M(j+1) - 1]} \\
\alpha^{[Mj + 1]} & \alpha^{[Mj + 2]} & \ldots & \alpha^{[M(j+1)]} \\
\vdots & \vdots & \ddots & \vdots \\
\alpha^{[M(j+1) - 1]} & \alpha^{[M(j+1)]} & \ldots & \alpha^{[M(j+2) - 2]} \\
\end{array} \right] \in \mathbb{F}_{q^m}^{M \times M},
\label{eq def T_j}
\end{equation}
for $ j = 0,1,2, \ldots, L $, where $ \alpha^{[i]} = \alpha^{q^i} $, for $ i \in \mathbb{N} $. Finally, define the non-catastrophic $ (N,k) $ convolutional code $ \mathcal{C} \subseteq \mathbb{F}_{q^m}[D]^N $ as that with polynomial parity-check matrix 
$$ H = (A,B) \in \mathbb{F}_{q^m}[D]^{(N-k) \times N}, $$
$$ A = \sum_{j=0}^\nu A_j D^j \in \mathbb{F}_{q^m}[D]^{(N-k) \times (N-k)} \quad \textrm{and} \quad B = \sum_{j=0}^\nu B_j D^j \in \mathbb{F}_{q^m}[D]^{(N-k) \times k}, $$
where $ \nu = \delta / (N-k) $, $ A_0 = I_{N-k} $, and $ B $ can be given from $ A $ by the rule
$$ A^{-1}B = \sum_{j=0}^\infty T_j D^j \in \mathbb{F}_{q^m}(\!(D)\!)^{(N-k) \times k}. $$

The following theorem combines \cite[Th. 3.1]{stronglyMDS} with \cite[Th. 5]{mahmood-convolutional}.

\begin{theorem}
The $ (N,k) $ convolutional code $ \mathcal{C} \subseteq \mathbb{F}_{q^m}^N $ described above is non-catastrophic, has degree $ \delta $ and is $ L $-MSRD for any sum-rank length decomposition of $ N $.
\end{theorem}

\end{document}